\RequirePackage{fix-cm}
\documentclass[twocolumn]{svjour3preprint}          
\smartqed  
\usepackage{graphicx}
\usepackage{color}
\usepackage{amsmath} 

\usepackage{amsfonts}

\usepackage{amsthm} 

\usepackage[round]{natbib} 

\usepackage{algorithmic}
\usepackage[mathscr]{euscript}

\usepackage{float}

\usepackage{comment}

\def\Var{\mathop{\rm Var}}

\def\argmin{\mathop{\rm Arg\,Min}\limits}
\def\argmax{\mathop{\rm Arg\,Max}\limits}
\newcommand{\e}[1]{\mbe\brac{#1}}

\newcommand{\mtc}{\mathcal}
\newcommand{\mbf}{\mathbf}

\newcommand{\wh}[1]{{\widehat{#1}}}

\newcommand{\ind}[1]{{\mbf{1}\{#1\}}}
\newcommand{\paren}[1]{\left(#1\right)}
\newcommand{\brac}[1]{\left[#1\right]}

\newcommand{\set}[1]{\left\{#1\right\}}
\newcommand{\abs}[1]{\left\lvert #1 \right\rvert}

\def\cF{{\mathscr{F}}}
\def\cG{{\mtc{G}}}

\def\cN{{\mtc{N}}}

\newcommand{\eps}{\varepsilon}

\newcommand{\mbe}{\mathbb{E}}
\newcommand{\mbr}{\mathbb{R}}
\newcommand{\mbp}{\mathbb{P}}

\newcommand{\mbn}{\mathbb{N}}

\newcommand{\prob}[1]{\mbp\brac{#1}}

\usepackage{color}

\newcounter{nbdrafts}
\makeatletter
\newcommand{\checknbdrafts}{
\ifnum \thenbdrafts > 0
\@latex@warning@no@line{**********************************************************************}
\@latex@warning@no@line{* The document contains \thenbdrafts \space draft note(s)}
\@latex@warning@no@line{**********************************************************************}
\fi}

\makeatother

\newcommand{\pdl}{p_L(d)}

\newcommand{\fre}{\mathsf{Fr\acute{e}chet}}

\journalname{Statistics and Computing}
\begin{document}

\title{Extensions of stability selection using subsamples of observations and covariates
\thanks{
A preliminary version of this work was presented
at the conference DAGM 2012 \citep{beinrucker2012simple}.}}

\subtitle{}

\author{Andre Beinrucker \and \"Ur\"un Dogan \and Gilles Blanchard
}

\authorrunning{A. Beinrucker, \"U. Dogan, G. Blanchard} 

\institute{Andre Beinrucker, Gilles Blanchard \at
           University of Potsdam\\
              Am Neuen Palais 10, 14469 Potsdam, Germany\\
              \email{\{andre.beinrucker,gilles.blanchard\}@uni-potsdam.de}
           \and                    
           \"Ur\"un Dogan  \at
           Microsoft/Skype Labs\\
              2 Waterhouse Square, 140 Holborn EC1N2ST London\\  
              Unided Kingdom\\
              \email{urundogan@gmail.com}              
}


\maketitle

\begin{abstract}
We introduce extensions of stability selection, a method to stabilise variable selection methods introduced by Meinshausen and Bühlmann (J R Stat Soc 72:417-473, 2010). We propose to apply a base selection method repeatedly to 
random subsamples of observations and subsets of covariates 
under scrutiny, and to select
covariates based on their selection frequency. We analyse
the effects and benefits of these extensions. Our analysis generalizes the
theoretical results of Meinshausen and Bühlmann (J R Stat Soc 72:417-473, 2010) from
the case of half-samples to subsamples of arbitrary size.
We study, in a theoretical manner, the effect of taking random covariate subsets
using a simplified score model. Finally we validate these
extensions on numerical experiments on both synthetic and
real datasets, and compare the obtained results in detail
to the original stability selection method.

\end{abstract}

\section{Introduction}

\subsection{Motivation}

In many applications a very large number of covariates are observed, of which only a few carry information about an outcome of interest. Variable selection techniques aim at identifying such relevant covariates (for a review see \citealp{guyon2006feature}). Usually, variable selection aims at one of two goals: to identify informative covariates in order to get scientific insight into the data and the process that generated the outcome; or to use the covariates identified as relevant in order to predict the outcome. 
In this work we primarily focus on the identification of informative covariates but also consider
prediction results using real data.
We consider {\em variable selection} (also called {\em feature selection} in computer science-related communities) as a part of the broader field of dimensionality reduction.

Many variable selection methods share the common drawback of being  unstable with respect to small changes of the data: if one estimates the set of relevant covariates on different sets of observations coming from the same source, the result can vary significantly. While this is not necessarily of concern if prediction is the goal, it makes the identification of relevant covariates very difficult. One approach to overcome this problem is {\em stability selection} \citep{meinshausen2010stability}. It consists in applying repeatedly a variable selection method to randomly chosen subsamples of half size of the observations. 
The final selection is obtained by picking only those covariates whose selection frequency across
repetitions exceeds a certain threshold. This threshold can be chosen such that (under some
assumptions) the expected number of false positive selections is guaranteed to be below a chosen value. 

\subsection{Contributions}

In the remainder of the paper we will refer to the variable selection method that is repeatedly applied to data subsamples as the {\em base method}. Similarly to \citet{meinshausen2010stability}, we construct a method that can be
applied on top of an arbitrary base method, which is considered as a black box.

We propose to extend the central idea of stability selection in two natural directions.
First,  \citet{meinshausen2010stability} use random samples containing half of the observations of the full dataset. Instead, we choose some integer $L>1$ and draw subsamples of size $\frac{1}{L}$ of the full sample size. More precisely, we randomly partition the observations into disjoint subsamples, extending the approach of complementary pairs stability selection - CPSS \citep{shah2013variable}. 
We investigate theoretically the behaviour of the expected number of false positive selections depending on the number of subsamples. In addition, we perform extensive simulation studies to compare the number of correct variables recovered for artificial and semi-synthetic datasets.

Secondly, \citet{meinshausen2010stability} remarked  from empirical comparisons that stability selection can be improved by randomising the base method. We propose a randomization by simply applying the base method to random  subsets of covariates. We obtain these subsets by randomly partitioning the covariates into disjoint subsets. 
Because the effect of doing so depends on the base method being used, it is difficult to analyse it
theoretically in full generality. In this work, we restrict our theoretical analysis to a simplified toy model, in which we assume that for each covariate there exists a latent score reflecting its informativeness about the outcome of interest. Furthermore, we assume that the base method has access to noisy observations of these scores, and outputs the covariate with the largest observed score. 
We investigate how the probability of selecting a noninformative covariate (false positive)
is influenced by the size of the random subset of covariates used.
Besides the theoretical analysis of this toy model, we performed simulation studies similar to the ones in the investigation of the subsampling of observations.

We call the method that combines the two extensions proposed {\em extended stability selection}. 
To summarize, it applies the base method repeatedly to randomly chosen subsets of the observations and covariates and finally ranking covariates by their selection frequency.
There are two parallel goals for this extension. The first goal is to improve the precision of the selection, that is to reduce the number of false positives. The second goal is to reduce the computational complexity of stability selection. 
Indeed, each call of the base method is restricted to a subset of observations {\em and}  covariates;
this reduces the memory requirements of the method. If the complexity of the base method grows faster than linearly in the data size (number of observations times number of covariates), the total computation cost is also reduced.
This is particularly appealing if the base method needs to load the data that it operates on into memory, which can be infeasible for large data matrices, but easy for smaller submatrices. Furthermore, this naturally allows for parallelization of the method, since these submatrices can be processed independently. 

\subsection{Overview of results}
Concerning the subsampling of observations, we obtain a bound on the expected number of false positives, depending on the size of the subsamples. This bound sharply generalizes Theorem 1 of \citet{meinshausen2010stability}
and Theorem 1 of \citet{shah2013variable}. Our results suggest that there is a trade-off between improving the selection of covariates for each individual subsample by using a smaller number of larger subsamples, and improving the final selection by averaging over a larger number of smaller, independent subsamples. This finding is in line with general insights on subsampling
methods \citep{politis1999subsampling} and cross-validation \citep[Section 10.3]{arlot2010survey}.
Even though our empirical comparison shows only small differences, a significant 
advantage of our proposed subsampling extension is that it has much less computational and memory requirements compared to the original stability selection or CPSS.

For the randomization of the base method obtained by taking disjoint subset of covariates, the
theoretical analysis of our simplified score model shows that under certain assumptions, there exists an optimal size for the randomly chosen covariate subsets. Our empirical results support this finding: such randomization generally improves the performance, unless the subset size is too small. 

\subsection{Organization of the paper}
In Section 2 we give a detailed description of the algorithm proposed, including 
the base methods considered in the experiments.
The theoretical analysis is presented in Section 3. We motivate the use of small observation subsamples in Section 3.1 and investigate the randomization of the base method in Section 3.2. Experimental results are given in Section 4, where we measure the performance of the algorithm in selecting informative covariates in Section 4.1 and 4.2 and apply our method in an image classification setting in Section 4.3. We conclude our work in Section 5 with a summary and a discussion.

\section{Methods} 

\subsection{Description of the algorithm} \label{sec:description}

In the sequel we assume to be given a dataset $\mathscr{D}$ containing $N$
observations $(X^{(i)},Y^{(i)})_{i=1,\ldots,N}$, each observation consisting of $D$ covariates $(X_1,\ldots,X_D)$
and an outcome $Y$ of interest. We choose $T$, the number of times we repeat the random partitioning of the data and a threshold $\tau \in (0,1)$ that indicates the fraction of observation subsamples in which a covariate needs to be chosen in order to enter the final selection. The number of observations and covariates that we apply our base method on is determined by the parameters $(L,V)$. 
We use random observation subsamples of size $\left\lfloor \frac{N}{L} \right\rfloor$ and  covariate subsets of size $\left\lfloor \frac{D}{V} \right\rfloor$ or $\left\lfloor \frac{D}{V} \right\rfloor + 1$.

We denote by $S^{\mathrm{base}}(X^{(\mathscr{S})}_{\cF},Y^{(\mathscr{S})})$ the output of the base method applied to the restriction of the data $\mathscr{D}$ to observations with indices $\mathscr{S}\subset\set{1,\ldots,N}$ and covariates with indices $\cF\subset\set{1,\ldots,D}$. We denote by $\Pi_{L,V}^{\mathrm{SFS}}(d)$ the selection frequency of covariate $d$ where the superscript SFS stands for {\em stability feature selection}. 
We give the pseudo-code of the method we propose below.  Note that we recover the original stability selection algorithm (more precisely, the CPSS algorithm of \citealp{shah2013variable}) for $L=2$, $V=1$.

\begin{algorithmic}
\STATE {\bf Parameters:}\\
$\bullet$ Number of iterations $T$ \\
$\bullet$ Number of observation subsamples per iteration $L$\\
$\bullet$ Number of covariate subsets per iteration $V$ \\
$\bullet$ Threshold $\tau\in(0,1)$ \\
\STATE {\bf Input:} $\mathscr{D} = (X^{(i)},Y^{(i)})_{i=1,\ldots,N}$ (with $X^{(i)} \in \mathbb{R}^D$.)
	\STATE {\bf Initialization:} selection frequencies $\Pi_{L,V}^{\mathrm{SFS}}(d)=0$, $d=1,\ldots,D$
	\FOR{$t = 1$ \TO $T$}
		\STATE Draw $L$ disjoint random subsamples $\mathscr{S}{(1)}, \ldots,\mathscr{S}{(L)}$ of size 
		$\left\lfloor \frac{N}{L} \right\rfloor$ without repetition of $\set{1,\ldots,N}$\,.
		\STATE Partition  $\set{1,\ldots,D}$ into $V$ disjoint random subsets $\cF{(1)}, \ldots, \cF{(V)}$ of size $\left\lfloor \frac{D}{V} \right\rfloor$ or $\left\lfloor \frac{D}{V} \right\rfloor + 1$\,.
		\FOR{$i=1,\ldots,L$; $j=1,\ldots,V$}
			\STATE $ \cG := S^{\mathrm{base}}\left(X^{(\mathscr{S}(i))}_{\cF(j)},Y^{(\mathscr{S}(i))}\right)$
			\FORALL{$d \in \cG$} 
	 			\STATE $ \Pi_{L,V}^{\mathrm{SFS}}(d) \leftarrow \Pi_{L,V}^{\mathrm{SFS}}(d) +1/LT$
			\ENDFOR
		\ENDFOR
	\ENDFOR
\RETURN set of indices in final selection \\
$\quad \quad \quad \quad S_{L,V,\tau}^{SFS}:=\{d: \Pi_{L,V}^{\mathrm{SFS}}(d) \geq \tau\}$.
\end{algorithmic}

\subsection{Comparison to previous work}

Statistical methods can be applied to subsamples of the data in various ways. A classical way in this context is the bootstrap \citep{efron1979bootstrap}, where subsamples are drawn with replacement. In contrast, stability selection and our extension follow the idea of subsampling without replacement \citep{politis1999subsampling} and are strongly related to subagging \citep{buhlmann2002analyzing}.

Several approaches have been developed to combine variable selection and subsampling of observations. \citet{sauerbrei1992bootstrap} investigated bootstrapping variable selection methods in the Cox regression model. Further, several methods have been proposed where a predictor that does variable selection intrinsically is applied to subsamples of the data. The selection obtained is then used as input to the final prediction method. \citet{bi2003dimensionality} used linear a support vector machines (SVM) to select variables and train a kernel SVM on them. Similarly, in Bolasso \citep{bach2008bolasso}, the Lasso method is  applied to bootstrapped subsamples of the data. The final selection is obtained as intersection of the sets of selected covariates on the bootstraps. On this final selection, Lasso is applied again to obtain a final estimate of the coefficients. Similarly, in random Lasso \citep{wang2011random} the importance of covariates is estimated by applying Lasso to bootstrap subsamples of the data. Afterwards Lasso is applied to a random selection of the covariates, the probability to be included is proportional to the measure obtained in the first step. 

\citet{meinshausen2010stability} proposed stability selection, a method that combines subsampling with variable selection. In their analysis the authors give a bound for the number of false positives under
some simplifying assumptions. Their work was the basis for numerous follow-up studies by theoreticians and practitioners.

\citet{shah2013variable} introduced complementary pairs stability selection (CPSS), a variant of stability selection which uses not only subsamples of size $\left\lfloor \frac{N}{2} \right\rfloor$ of the data, but also their complements. They loosen the assumptions of the original method and give bounds on both errors of the selection procedure, false positives and false negatives. Our work is an extension of their results, as we apply our base method to complementary subsamples as well.

Stability selection has been applied in various disciplines. Fields of application include genome-wide association studies \citep{alexander2011stability, he2011variable}, biomarker discovery \citep{he2010stable} and inference of gene regulatory networks  \citep{haury2012tigress}. 

The idea of repeatedly applying a statistical method to covariate subsets has been investigated before. One famous example is Random Forest \citep{breiman2001random} where decision trees are build on covariate subsets. Each decision tree can be regarded as a variable selection method as well. Recently, \citet{hinton2012improving} remarked that omitting randomly chosen covariates in the training of a neuronal network improves classification accuracy on test data drastically. A different approach to reduce the dimensionality of the problem is to cluster the covariates and apply a variable selection method to cluster representatives \citep{buhlmann2013correlated}.

\subsection{Base methods}\label{sec:basemethods}
In stability selection any variable selection method can in principle be used as a base method. In this section, we describe variable selection methods in general and give details about two selection methods that we used as a base method for our extended stability selection experiments. Base methods are described below 
for the full sample $\mathscr{D}$ for simplicity.

Variable selection methods can be classified into filters, wrapper and embedded methods \citep{guyon2006feature}. Computationally most efficient are usually filter methods, which perform variable selection independently of the specific statistical treatment that might be applied afterwards. Examples are methods based on simple univariate correlation between covariates and the outcome of interest or mutual information (see, e.g., \citealp[Chapter 2]{cover2006elements}). 
Wrapper methods evaluate the relevance of a subset of covariates using the output of the ensuing statistical treatment (typically regression or discrimination) computed on the subset of covariates only.
Embedded methods perform variable selection and prediction simultaneously.

To assess the performance of the proposed methods, we choose two popular methods as base method: CMIM - conditional mutual information maximisation \citep{fleuret2004fast} from the class of filters and Lasso \citep{tibshirani1996regression} from the class of embedded methods. Lasso was also used as a base method in the original work of \citet{meinshausen2010stability}

\subsubsection{CMIM - conditional mutual information maximization}
\label{sec:CMIM}
Intuitively, the aim of the CMIM algorithm \citep{fleuret2004fast} is to find a subset of covariates of given size $K$ that maximizes the amount of information that the selected covariates $X_{\nu(1)}, \dots , X_{\nu(K)}$ contain about the outcome $Y$. It usually finds a good trade-off between relevance and redundancy of the selected covariates and is much faster than many competing mutual information based variable selection methods.

To make the computation feasible, the CMIM algorithm does not look directly for the set of covariates that globally maximizes the mutual information with the target, but performs 
greedy selection by iteratively selecting the covariate that has the largest mutual information with the target, conditional to the set of already selected covariates. Furthermore, the latter quantity is approximated by a simpler upper bound, namely, the minimal information about the target that a candidate covariate adds to any of the already selected covariates, taken individually (rather than jointly).
The final algorithm takes the following simple form:
\begin{align}
  \nu(1) & = \argmax_{d\in\set{1, \dots, D}} \wh{I}(X_d;Y)\,; \label{eq:cmim1} \\
  \nu(\ell) & = 
   \argmax_{d\in\set{1, \dots, D}} \label{eq:cmim2} 
   \min_{j < \ell} \wh{I}\left(X_d;Y|X_{\nu(j)} \right)\,, \quad \ell >1\,,
\end{align}
where $\wh{I}(X;Y)$ denotes an estimator of the mutual information of $X$ and $Y$ and $\wh{I}(X;Y|Z)$ an estimator of the mutual information between $X$ and $Y$ given $Z$. To speed up the computation, one can use a fast implementation of the algorithm \citep{fleuret2004fast}.

\subsubsection{Lasso}
The Lasso problem \citep{tibshirani1996regression} is to find an $\ell_1$ regularized solution for the least squares problem in the linear model. It can be stated as
\begin{equation}
\label{eq:lasso}
  \hat{\beta}_{\lambda} = \argmin _{\beta \in \mbr^D} \quad \sum_{i=1}^N (Y^{(i)}- \langle X^{(i)},\beta \rangle)^2 + \lambda |\beta| \,,
\end{equation}
where $\lambda \in \mbr_{>0}$ is some regularization parameter.

Thanks to the geometric properties of the $\ell_1$-Norm, solutions of the Lasso problem tend to be sparse, meaning that only a few coefficients of $\beta$ are non-zero. This property allows us to use Lasso solvers as variable selection methods, and to define 
for any $\lambda >0$ the set of selected covariates as the ones corresponding to non-zero coefficients.
Additionally to the selection of relevant covariates, a solution $\hat{\beta}_{\lambda}$ of the Lasso problem can be used to predict an outcome for new data. Examples for popular Lasso solvers are LARS \citep{hastie2012lars} and GLMNET \citep{friedman2010regularization}.


\section{Analysis}
\label{se:anal}

In this section, we analyse the effect of the size of the observations subsamples as well as the effect of taking random covariate subsets on the performance of the final method. 

\subsection{Subsampling of observations}
\label{se:subobs}

In this section the $N$ observations of the full sample $\mathscr{D}$ are always assumed to be drawn i.i.d. from an underlying, unknown generating distribution. 
As we do not take covariate subsets in this section, we fix $V=1$ and suppress the dependence on $V$ 
in the notation.  We recall that we repeat $T$ times the random draw of $L$ disjoint observation subsamples 
of equal size without repetition; we denote $\mathscr{S}(\ell,t)$ the 
$\ell$-th subsample of observations indices in the $t$-th drawing.
 On the $L\times T$ subsamples that we obtain, 
the selection frequency of covariate $d$ is
\[
  \Pi^{SFS}_L(d):= \frac{1}{TL} \sum_{t=1}^{T}\sum_{\ell=1}^L 
  \ind{d \in S^{\mathrm{base}}(X^{(\mathscr{S}(\ell,t))},Y^{(\mathscr{S}(\ell,t))})}\,.
\]
Thresholding this quantity, we obtain the output of the stability selection procedure. For any $\tau\in(0,1)$ we define
\[
S^{SFS}_{L,\tau}:= \set{d: \Pi^{SFS}_L(d) \geq \tau}.
\]

To evaluate the performance of the method, we need to define the set of covariates that we would like to be excluded from our selection, i.e. the covariates that we consider false positives if they are selected. Since the base method is treated
here as a black box and otherwise unspecified, we have to trust that, {\em on average} (over a random i.i.d. sample),
the base method selects relevant covariates more frequently than irrelevant ones.
Consider virtually drawing an independent, i.i.d. set of observations of size $\left\lfloor \frac{N}{L} \right\rfloor$ and denote the random output of the base method on this sample by $S^{\mathrm{base}}_L$.
Then we define for each covariate $d$ 
\begin{equation*}\label{eq:pdl}
\pdl:=\mathbb{P}\left[d\in S^{\mathrm{base}}_L\right],
\end{equation*}
its probability to be selected by the base method using a sample of size $\left\lfloor \frac{N}{L} \right\rfloor$.

The quantities $\pdl$ give us a yardstick to rank the covariates by relevance; 
accordingly, for any given threshold $\theta\in(0,1)$ we define
\[
  A_{\theta,L} := \{d:\pdl\le \theta\}
\]
the set of uninformative covariates at base selection probability lower than $\theta$.
In this definition, and the assumption that the probability of selection under the base method reflects
the true relevance, we follow the general approach of \citet{shah2013variable};  the relation to the assumptions of \citet{meinshausen2010stability} is discussed below.
Observe that, since each subsample appearing in the definition $\Pi^{SFS}_L(d)$ is individually
an i.i.d. sample of size $\left\lfloor \frac{N}{L} \right\rfloor$, it follows that
$\Pi_L^{SFS}(d)$ is an unbiased estimate of $\pdl$.

The following theorem bounds in Equation \eqref{eq1} the ratio of the expected number of false positive selections of $S_{L,\tau}^{SFS}$ compared to the total number of uninformative covariates (false positive rate), as well as compared to the expected number of false positive selections of the base method applied on a single subsample in Equation \eqref{eq2}. A corresponding result for the false negatives is available in Equations \eqref{eq3} and \eqref{eq4}.

\begin{theorem}\label{th:sampleSplitting}
  Let $L\in \mathbb{N}$ 
  and $\tau\in (0,1)$. Denote 
  $\forall p,q\in (0,1) \, D(p,q):= p\log \frac{p}{q} + (1-p) \log \frac{1-p}{1-q}$ the Kullback-Leibler
  divergence between two Bernoulli variables of parameters $p$ and $q$.
  For two integers $a\leq b$ denote $\set{a..b}$ the integer interval with endpoints $a,b$.
  Depending on the choice of $\theta$ and $\tau$ we have two cases with two results each:

  If $\theta<\tau$ we have
  \begin{align}\label{eq1}
	&\frac{\e{\abs{S^{SFS}_{L,\tau} \cap A_{\theta,L}}}} {\abs{A_{\theta,L}}} \\ 
    &\quad \leq
    \min_{\ell_0 \in \set{\lceil L\theta \rceil .. \lceil L \tau \rceil}}
    \paren{\frac{L-\ell_0 +1}{\tau L - \ell_0 + 1}}\exp\paren{ - L D\paren{\frac{\ell_0}{L},\theta}} \nonumber 
  \end{align}
  and
  \begin{align}\label{eq2}
    &\frac{\e{\abs{S^{SFS}_{L,\tau} \cap A_{\theta,L}}}}{\e{\abs{S^{\mathrm{base}}_{L} \cap A_{\theta,L} }}} \\
    &\quad \leq \min_{\substack{\ell_0 \in\\ \set{ \lceil L\theta \rceil + 1 .. \lceil L \tau \rceil}}}
    \paren{\frac{L-\ell_0 +1}{\tau L - \ell_0 + 1}}
    \frac{\exp\paren{ - L D\paren{\frac{\ell_0}{L},\theta}}}{\theta} \,,\nonumber
  \end{align}

  Similarly, if $\tau<\theta$ we have
  \begin{align}\label{eq3}
    &\frac{\e{\abs{\paren{S^{SFS}_{L,\tau}}^c \cap \paren{A_{\theta,L}}^c}}}{\abs{\paren{A_{\theta,L}}^c}}\\
    &\quad\leq \min_{\substack{\ell_0 \in \\ \set{\lfloor L\tau \rfloor .. \lfloor L \theta \rfloor}}}
    \paren{\frac{\ell_0 + 1}{\ell_0 - \tau L + 1}} \exp\paren{ - L D\paren{\frac{\ell_0}{L},\theta}}\,,\nonumber
  \end{align}
  and 
  \begin{align}\label{eq4}
    &\frac{\e{\abs{\paren{S^{SFS}_{L,\tau}}^c \cap \paren{A_{\theta,L}}^c}}}{\e{\abs{\paren{S^{\mathrm{base}}_L}^c\cap\paren{A_{\theta,L}}^c}}}  \\
    &\quad\leq \min_{\substack{\ell_0 \in\\ \set{\lfloor L\tau \rfloor .. \lfloor L \theta \rfloor -1}}}
    \paren{\frac{\ell_0 + 1}{\ell_0 - \tau L + 1}}
    \frac{\exp\paren{ - L D\paren{\frac{\ell_0}{L},\theta}}}{1-\theta},\nonumber  
  \end{align}
  where $A^c$ denotes the complement of a set $A$ (in $\set{1,...,D}$).

\end{theorem}
For the special case $L=2$ we recover the results of \citet[Theorem 1]{shah2013variable} by choosing $\ell_0=2$ in Equation \eqref{eq1} and $\ell_0=0$ in \eqref{eq3}. In the following corollary we formulate our results under the assumptions and in the notation of \citet[Theorem 1]{meinshausen2010stability}:

\begin{corollary}\label{co:sampleSplitting}
Suppose we are given a set of noise covariates $\cN$ and a set of signal covariates. Assume that all noise covariates have the same probability to be selected by the base method.
Assume further that the base variable selection method has a larger probability to select any informative covariate than random guessing. We denote $q_L=\e{\abs{S^{\mathrm{base}}_L}}$.
Then for any $\tau > \frac{q_L}{D}$:
\begin{align}
\label{eq:MB-like}
&\frac{\e{\abs{S^{SFS}_{L,\tau} \cap \cN}}}{|\cN|} \\
&\quad\leq \min_{\substack{\ell_0 \in\\ \set{ \left\lceil \frac{L q_L}{D} \right\rceil .. \left\lceil \tau L \right\rceil}}}
\paren{\frac{L-\ell_0+1}{\tau L - \ell_0 + 1}}
\exp\paren{ - L D\paren{\frac{\ell_0}{L},\frac{q_L }{D}}}\,.\nonumber 
\end{align}
\end{corollary}
If we choose $\ell_0=L=2$ and we use that \linebreak $\exp(-2 D(1,\theta)) = \theta^2$ we recover the order $O(q_L ^2/D)$ of the bound  of \citet[Theorem 1]{meinshausen2010stability} as well as the constraint $\tau > 1/2$ there.

Bound \eqref{eq:MB-like} involves a minimum over the allowed values of $\ell_0$,
which is merely a technical parameter in the bound. In order to make the bound more readable,
we can pick a specific value of $\ell_0$ as follows.
As the exponential term in the bound is monotonically decreasing in $\ell_0$, the largest allowed value 
$\ell_0=\lceil \tau L \rceil$ seems a natural choice.  
However, one should ensure that 
the multiplicative term in front does not become too large.
For this, 
choose $\tau$ and $L$ such that $\tau L$ is an integer; then we have $\frac{\ell_0-1}{L} = \frac{\tau L -1}{L} = \tau - \frac{1}{L}$. In this case, Corollary \ref{co:sampleSplitting} simplifies as follows.

\begin{corollary}
  Under the assumptions of Corollary \ref{co:sampleSplitting}, if $\tau L$ is an integer and if we choose $\ell_0= \tau L $, then we have as a direct consequence of Corollary \ref{co:sampleSplitting} for any $\tau > \frac{q_L }{D}$:
\begin{equation}
\frac{\e{\abs{S^{SFS}_{L,\tau} \cap \cN}}}{|\cN|} 
\leq \left( L(1 - \tau) + 1 \right)
\exp\paren{ - L D\paren{\tau,\frac{q_L }{D}}}\,.
\end{equation}  
\end{corollary}

\begin{remark}.
  \label{rem1}
As the expected number of false positive selections decays exponentially with $L$, it is tempting to conclude that $L$ should be chosen as large as possible. But one should not forget that the parameter $L$ has an important influence on base selection probabilities $\pdl$ as well.
This influence strongly depends on the particular base method used, and we cannot hope
to derive a generic quantitative 
statement concerning that point. Generally speaking, we expect that
as $L$ becomes larger and the subsample size smaller, the base method will receive less information and therefore will get closer to random guessing; in other words, we expect that
for larger $L$, base selection probabilities $\pdl$ are all pulled together closer to 
the value $q_L /D$ (corresponding random selection). 
Even assuming the ranking of the values of $\pdl$ is unchanged for different
values of $L$, the set $A_{\theta,L}^c$ of relevant variables at base selection probability
larger than $\theta$ will contain a smaller number of covariates for larger $L$ and
fixed $\theta>q_L /D$. To maintain the same number of covariates in this relevant set, one has to consider a lower value of $\theta$ for larger $L$. In other words, as $L$ increases,
there is a trade-off between variance reduction of the selection frequencies $\Pi_L^{SFS}(d)$ as quantified by Theorem \ref{th:sampleSplitting}, and the reduced separation of their means $\pdl$, both of which are important for successful discrimination of relevant covariates.  
A similar effect occurs in $L$-fold cross-validation as when $L$ grows the number of test sets available increases, but the size of each test set decreases. This trade-off has been discussed extensively, see for instance \citet[Section 10.3]{arlot2010survey}.


{\em Illustration of the bounds.} In Figures \ref{fig:bound} and \ref{fig:tauQ} we illustrate the bound given by Corollary \ref{co:sampleSplitting}.
On Figure~\ref{fig:bound}, we show the dependence of the bound on $\tau$ for $D=1000$, $q_L = 28$ and $L\in \{2,4,8\}$.
We see that, for a fixed error level guaranteed by the bound, for larger $L$ we can choose a smaller frequency selection threshold $\tau$ and therefore
possibly select more covariates (while keeping in mind the tradeoff effect
discussed in Remark~\ref{rem1}).
Moreover, we see that for error control ensuring a very low level of
false positives, the bound for $L=2$ is not applicable, while the bound for larger $L$ is. 

\begin{figure}[b]
  \centerline{\includegraphics[width=0.4\textwidth]{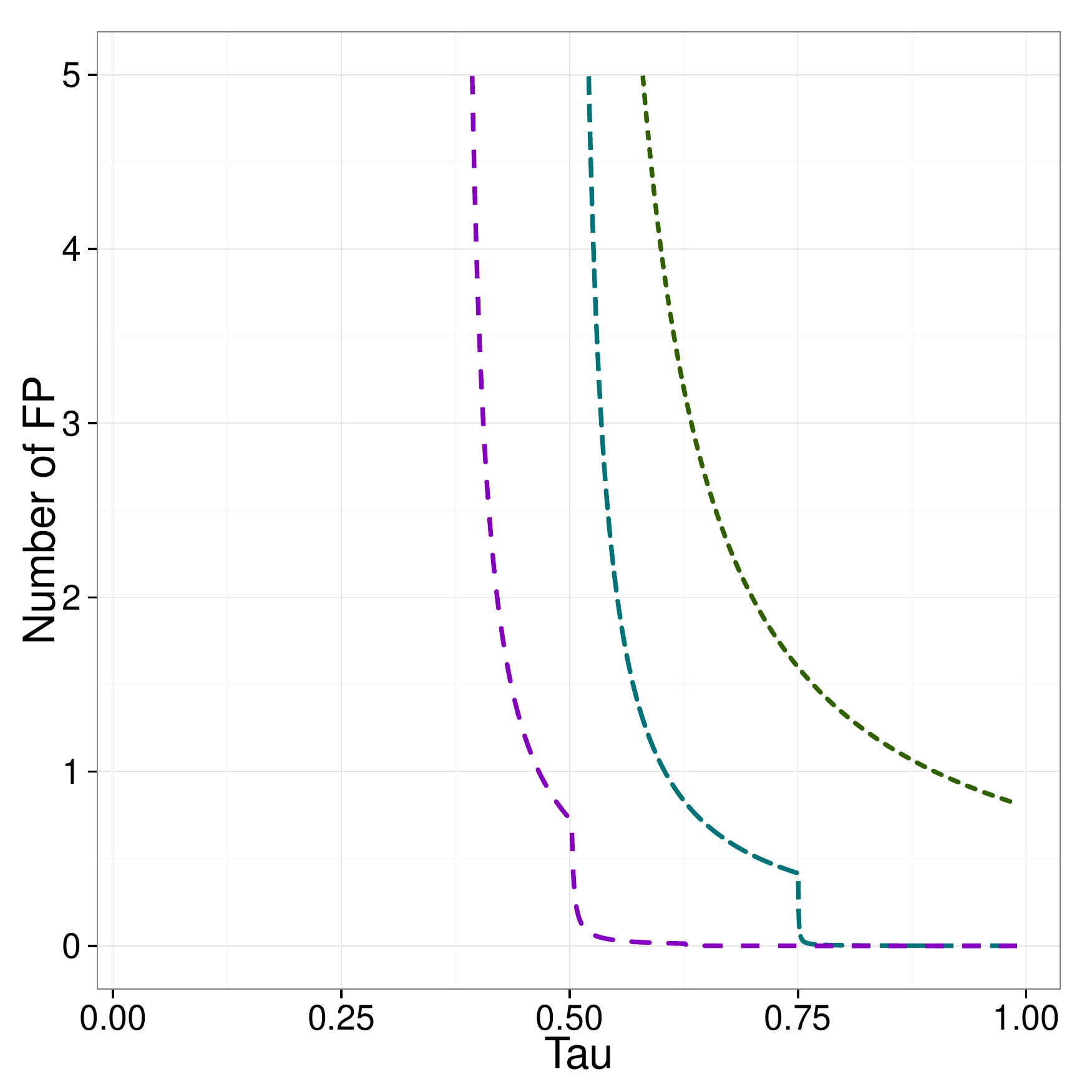}\includegraphics[height=6cm]{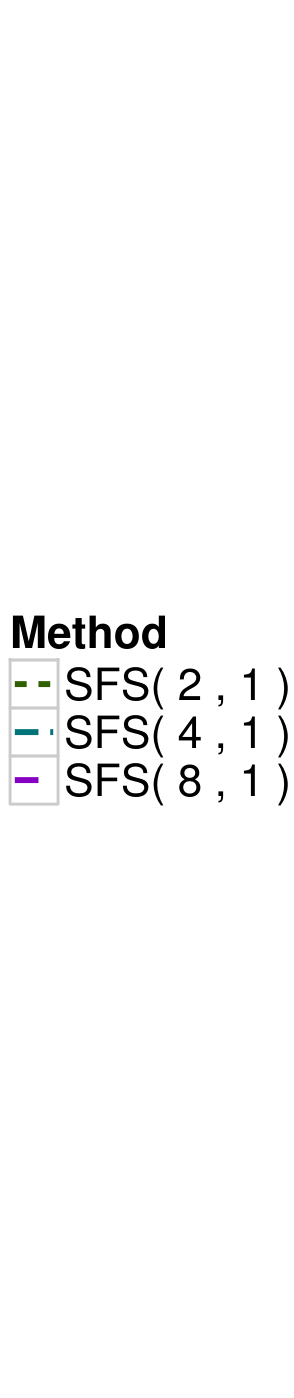}}
  \caption{Dependence of the bound of Corollary \ref{co:sampleSplitting} on $\tau$ for various $L$, $q_L=28$}
  \label{fig:bound}
\end{figure}

In Figure~\ref{fig:tauQ}, we fix the expected number of false positives that we can tolerate, and use for each $L\in \{2,4,8\}$ and each $q_L \in \{1,..,100\}$ Corollary \ref{co:sampleSplitting} to determine the 
smallest selection frequency $\tau_{\min}$ that guarantees this bound. On this figure, we
enforce at most one false positive on average, results for two and five FPs are given in the supplementary material. Reported is the behaviour of
$\tau_{\min}$ as a function $q_L$, the number of covariates selected by the baseline. Note that for standard stability selection ($L=2$) the bound in Corollary \ref{co:sampleSplitting} (which we recall coincides with the result of \citealt{meinshausen2010stability})
does not achieve the value one for any $q_L \ge 32$.

\begin{figure}[t]
  \centerline{\includegraphics[width=0.4\textwidth]{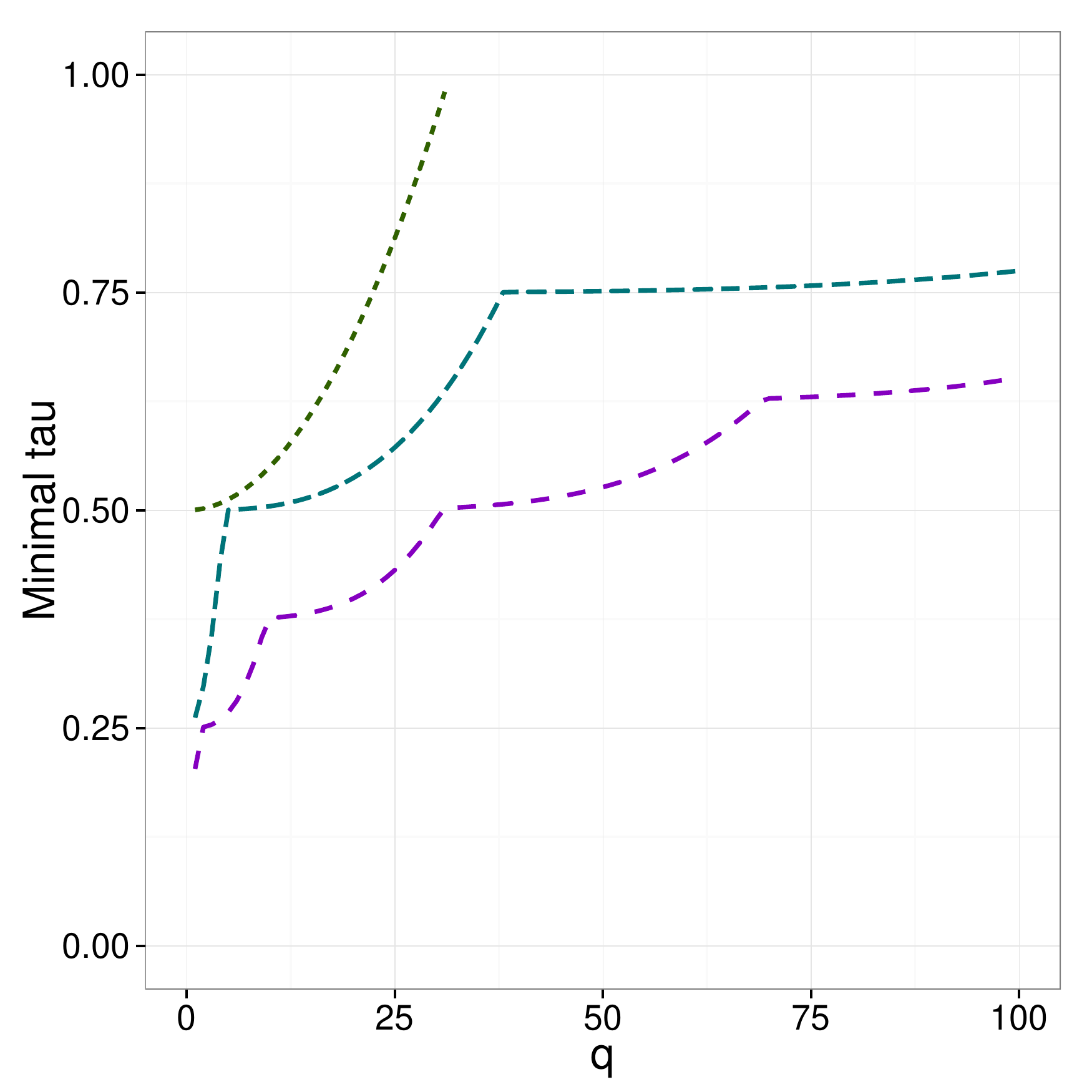}\includegraphics[height=6cm]{legendFp}}
  \caption{The minimal value of $\tau$ such that the bound in Corollary \ref{co:sampleSplitting} is below 1, as a function of $q_L$, for various $L$.}
  \label{fig:tauQ}
\end{figure}


An automatic choice of $L$ is a thorny theoretical question and ultimately depends on the intended goal. 
If the goal is to select truly informative covariates with a strict control of the number of false positives,
the above theoretical bounds can provide guidance; in Section~\ref{se:artif}, we investigate their
practical relevance to drive the choice of $L$ based on a given $q_L$ and a target bound on
the average number of false positives. If the goal is to improve prediction accuracy,
we recommend indirect assessment methods such as  cross-validation. If some constraints due to computing scalability
are present, we recommend choosing $L$ according to possible memory constraints or available parallel computing capabilities.
\end{remark}

\subsection{Randomization by taking covariate subsets}
\label{se:subcov}
It is not possible to study the effect of taking random covariate subsets on the selection probabilities in as much generality as we studied the effect of observation subsampling in the previous section. The reason is that this effect depends prominently on the specific base procedure used (see Section \ref{sec:basemethods}). In this section, 
we assume an iterative, score-based base selection procedure (such as CMIM, see Section \ref{sec:CMIM}). 
We further limit the analysis and zoom in on a single iteration of this procedure (thus considering only the
selection of 1 covariate) using a strongly simplified modelling of covariate scores.
Although the result of further iterations generally strongly depends on previous ones,
the prototypical model considered here for one iteration already highlights some interesting
behaviour of the covariate subsampling procedure.

The simplified model is as follows. We assume that each covariate has an underlying score $Q_d$ reflecting its true informativeness; only an estimation $\wh{Q}_d$ of  that score is available, which we assume to follow the simple additive model
\begin{equation}
  \hat{Q}_d = Q_d + \eps_d,  \qquad \qquad d = 1, \dots, D\label{eq:Q_hat}.
\end{equation}

In Section \ref{se:subobs} we used the probability $\pdl$ to be selected by the base method as a yardstick to tell signal from noise covariates. In that case, the quantity $\pdl$ is an example for a score $Q_d$.

We expect that the amplitude of the random estimation noise $\eps$ will typically depend on the size of the observations subsample and thus on the parameter $L$. However in this section we consider $L$ as fixed and therefore omit it from the notation from now on.

The base procedure then outputs the covariate with the largest estimated score.
We want to compare this base procedure to the {\em randomized base procedure} consisting in first picking at random a subset of $D' < D$ covariates, and returning out of those the one with the largest estimated score.

Similarly to what was considered in the previous section, we define uninformative covariates as those having true scores below a certain threshold $\theta$:
\begin{align}
  A_{D,\theta}:=\set{ 1\leq d \leq D: Q_d \leq \theta};\label{eq:A_D}
\end{align}
we also denote $A_{D,\theta}^c:=\set{1,\ldots,D}\setminus A_{D,\theta}$.
Denote $p(d)$ and $p^{rand}(d)$ the probability of selecting covariate $d$ using the deterministic and the randomized base procedure respectively.
It is desirable for these to be as large as possible for $d\in A^c_{D,\theta}$, so that we compare the two base procedures by means of the sum of these probabilities, i.e. , for the deterministic base procedure, $\sum_{d \in A_{D,\theta}^c} p(d)=\prob{\wh{d}_D \in A_{D,\theta}^c}$, where $\wh{d}_{D}$ denotes the index of the covariate selected by the deterministic base procedure.

In the following theorem we analyse the behaviour of the latter quantity as $D$ grows. The main theoretical finding of this section is that, under certain circumstances concerning the distribution of the estimation noise, in an asymptotic sense $\hat{d}_D$ will be determined only by the estimation error, and not by the true score. In other words, if $D$ grows too large, the deterministic selection resembles picking at random.

This therefore supports the principle of randomizing the base procedure by taking subsets of  covariates in stability selection, since when the total number of covariates $D$ is large enough, the  probability of correct selection will be higher when taking  a random covariate subset of size $D'<D$. This phenomenon is illustrated by
a small simulation example at the end of the present section.

\begin{theorem}\label{th:frechet}
Consider a sequence of models of the form \eqref{eq:Q_hat}, a fixed number $\theta$, and
the following assumptions: 
\begin{itemize}
\item[(i)] The true scores $Q_d$ belong to the bounded interval $[0,M]$.
\item[(ii)] The noise variables $\eps_i$ are independent and identically distributed,
and their marginal distribution belongs, for some $\alpha>0$, to the maximal domain of attraction (MDA) of
a $\fre(\alpha)$ distribution (for details see Chapter 3.2 and 3.3 in \cite{embrechts1997modelling}).
\item[(iii)]  As $D\rightarrow \infty$, $\frac{\abs{A_{D,\theta}}}{D} \rightarrow \eta \in [0,1]$, where
$A_{D,\theta}$ is defined by \eqref{eq:A_D}.
\end{itemize}

Then for $\hat{d}_D = \argmax_{d = 1, \dots, D} \hat{Q}_d$, it holds that 
  \begin{align*}
\lim_{D\rightarrow \infty}    \mathbb{P}\left[ \hat{d}_D \in A_{D,\theta} \right]   \rightarrow \eta.
  \end{align*}
\end{theorem}

{\em Interpretation of the theorem.} Consider for comparison the blind strategy of simply
drawing a covariate uniformly at random among $D$, regardless of observed scores.
Then the probability for this covariate  
to lie in $A_{D,\theta}$ is obviously $\abs{A_{D,\theta}} / D.$ 
Thus, the theorem states that as $D$ grows, the strategy
consisting in picking the largest observed scores will asymptotically 
not be any better than the blind
strategy.

{\em Comments on assumption (iii).} This assumption concerns the asymptotic repartition of the true scores $Q_d, \linebreak 1\leq d \leq D$,
as $D$ grows. It is quite weak, and allows for the family of true scores to depend on $D$, provided this assumption remains satisfied. 
In particular, we can apply the theorem if the true scores are themselves random, and such that the assumption is satisfied a.s.  
In that situation, the theorem can be applied conditionally to the true score sequence and the conclusion will hold (for conditional incorrect selection probabilities) 
for almost all realizations of this sequence, and therefore also
in expectation over the scores (i.e. for unconditional probabilities).

A simple and intuitive instance of the above is when  
the true scores $Q_i$ are themselves assumed to be i.i.d. random draws following some
(arbitrary) distribution on the interval $[0, M]$. In that
case, assumption (iii) is satisfied a.s. with $\eta = \prob{Q_1 \leq \theta}$
by the law of large numbers.
Additionally,  if the true scores are modelled as
random i.i.d., the randomized procedure is equivalent to the base
procedure with $D$ simply replaced by $D'$.
In this sense, the conclusion of the theorem applies to the randomized selection, as well.
Denoting $E_{D}:=\prob{\wh{d}_{D} \in A_{D,\theta}}$ in this setting, we have clearly 
$E_1=\prob{Q_1\leq \theta}:= \eta$ 
as well as $\lim_{D\rightarrow \infty} E_D=\eta$ by the theorem. On the other hand, it is easy to see that $E_D>\eta$ for any $D>1$ (for any $D>1$, selecting the covariate with largest observed score must be at least slightly better than random guessing).
We conclude that $E_D$ is not monotone in $D$, and that it must attain a minimum value for some finite $D_{opt}>1$. (In the simulations shown at the end of the section, we actually see that $E_D$ is unimodal.) The same applies to $\wh{d}^{rand}_{D'}$, and we conclude that if $D>D_{opt}$, then it brings an advantage to select covariates out of random subsets of size smaller $D'$ (the optimal size being $D'=D_{opt}$). To sum up the finding in an (apparently) paradoxical statement: as $D$ grows too large, the deterministic selection behaves more randomly than the randomized selection using $D'<D$.

{\em Comments on assumption (ii).} The independence assumption is needed to apply classical results of extreme value theory. It is arguably unrealistic, and made here in order to illustrate the phenomenon in the simplest setting possible. We note that some extreme value results are also available under weak dependence models \citep[Chapter 3]{leadbetter1983extremes}, so that this assumption might be relaxed somewhat, though this is out of the scope of the present work.

The assumption that the noise marginal distribution belongs to MDA(Fr\'echet)
is needed to apply classical results of extreme value theory. Without getting into the details, this assumption (roughly) means that the distribution is heavy-tailed. This family includes Cauchy, Student's t, Pareto, Burr and Loggamma distributed noise \citep[Table 3.4.2]{embrechts1997modelling}. This assumption is reasonable, for instance, if we consider that the estimated scores are based on $t$-statistics estimated from a limited number of observations. In contrast, the next result shows that the phenomenon described in Theorem \ref{th:frechet} does not occur for Gaussian distributed noise:

\begin{theorem}\label{th:gaussian}
Consider a sequence of models of the form \eqref{eq:Q_hat}, a fixed number $\theta$, and
the following assumptions: 
\begin{itemize}
\item[(i)] the noise variables $\eps_i$ are independent and identically distributed with normal distribution.
\item[(ii)] $\liminf_{D\rightarrow \infty} \frac{\abs{A^c_{D,\theta}}}{D} >0$, 
where $A_{D,\theta}$ is defined by Equation \eqref{eq:A_D}.
\end{itemize}
Then for any $\theta' < \theta$:
  \begin{align*}
\lim_{D\rightarrow \infty}    \mathbb{P}\left[ \hat{d}_D \in A_{D,\theta'} \right]   =0.
  \end{align*}
\end{theorem}

{\em Remark.} The above does not contradict existing results
on support recovery of Lasso (reducing to a simple score
thresholding procedure in the case of orthogonal design), which is of sample complexity
$O(\log(D))$ and thus eventually fails if the dimension $D$ grows too large, for a fixed sample size.
This is because such results study the exact recovery of all
    of the support covariates. By contrast, the result of Theorem 3 concerns the much weaker requirement of recovery (with probability close to 1) of one single informative covariate in a situation where the total proportion of informative covariates is non-vanishing.

We illustrate by a simulation study how the error probability $\mathbb{P}\left[\hat{d}_D \in A_{D,\theta} \right]$ depends on the distribution of the noise $\eps_d$. We simulated  10000 noisy scores according to Equation \eqref{eq:Q_hat} where $Q_d \sim 2*\mathrm{Bernoulli}(0.1)$ for each of the five noise distributions Gaussian, Cauchy, Student's t with 3, 5 and 10 degrees of freedom. We count how often the largest noisy score comes from an uninformative true score and plot this frequency against the total number of scores $D$ in Figure~\ref{fig:effect_D}. We see that for 
the Gaussian distribution, the probability to select a uninformative covariate is monotone decreasing to zero as $D$ increases. In contrast, for heavy tailed distributions such as Cauchy and Student's t, the error probability decreases until reaching a minimum and increases afterwards. 

\begin{figure}[b]
\centerline{\includegraphics[width=0.5\textwidth]{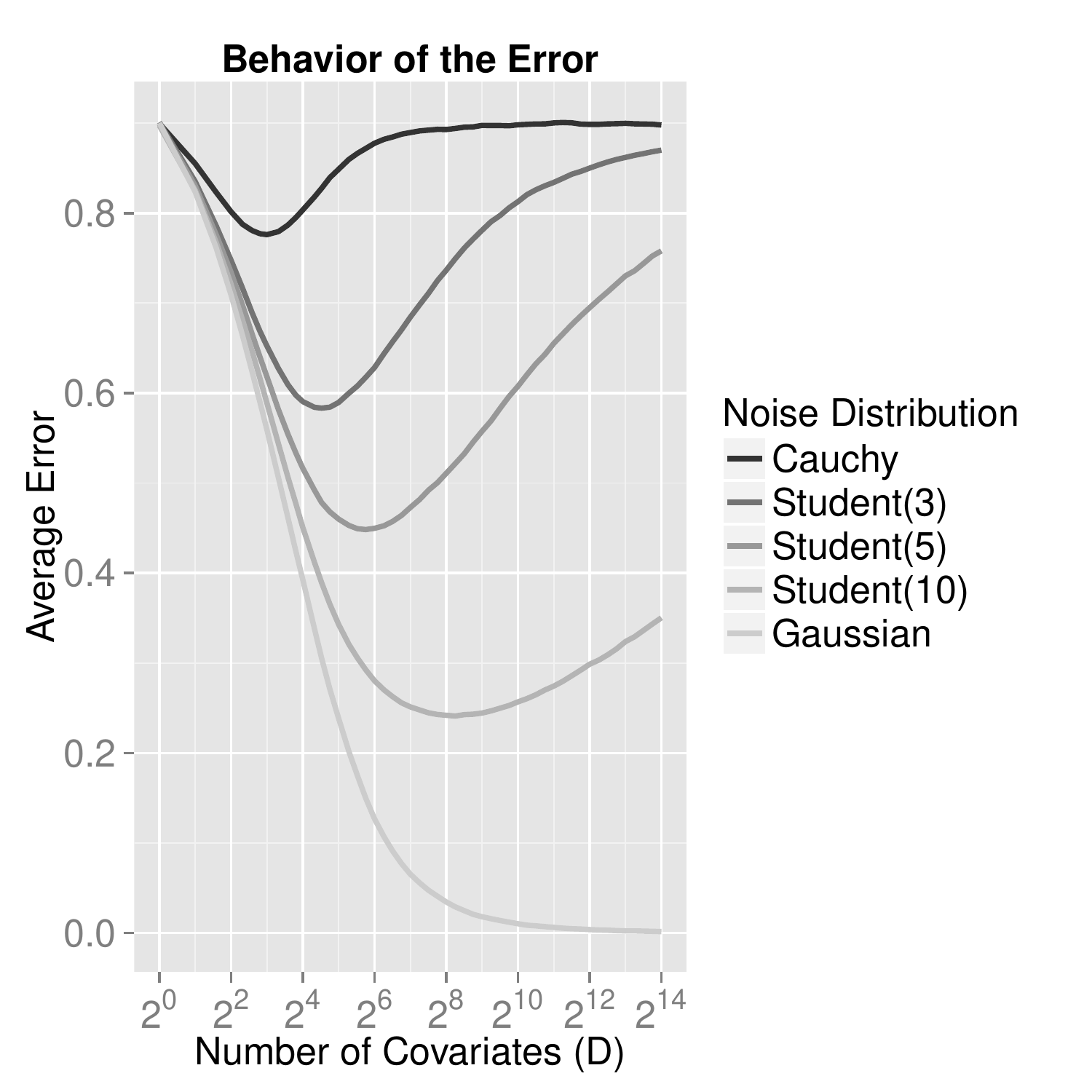}}
\caption{Dependence of the average error probability on $D$ for different noise distributions}
\label{fig:effect_D}
\end{figure}

We believe that the above results can be extended to the selection of $k$ covariates
$\left(\hat{d}^{(1)}_D,\ldots,\hat{d}_D^{(k)}\right)$ having the largest scores,
using results about the joint limit distribution of $k$ upper order statistics \citep[Theorem 4.2.8]{embrechts1997modelling}.
However, we emphasize that this does not appear to be a direct consequence of the above result,
since $\mathbb{P}\left[ \hat{d}^{(k)}_D \in A_{D,\theta} \right]$ is not necessarily monotone increasing in $k$.

\section{Experimental results}
In this section we evaluate the performance of our method using three different criteria and several datasets. For the first two criteria, we aim at the identification of informative covariates and consider the number of true and false positives selected in a controlled framework where the ground truth is known. In the third setting, we consider a real-data classification task and evaluate the effect of variable selection on the prediction performance of a learning algorithm using only the selected covariates.
\subsection{Identification of informative covariates}
\label{se:artif}

{\bf Setting.} 
For the first two criteria we generate the output variable of interest according to a known linear model (except for the Vitamin dataset):
\begin{equation}
  \label{linearModel}
  Y^{(n)}=\langle X^{(n)},\beta \rangle + \eps^{(n)}\quad \forall n=1,...,N\,,
\end{equation}
where $X^{(n)}$ and  $\beta$ are $D$-dimensional random vectors and $\eps^{(n)} \sim \mathcal{N}(0,1)$.

Additional experiments where $\eps^{(n)}$ follows a heavy-tailed Student(3) distribution are given in the supplemental material only, as the results are similar. The vector of coefficients $\beta$ contains only 20 non-zero entries. Their indices are chosen randomly and their values are generated from a $U[0,1]$ distribution. This setting is similar to the one considered
by \citet{meinshausen2010stability}.

We consider several different settings for the design matrix $X$, corresponding to 
controlled simulated situations or to real data. Except otherwise specified,
each experiment is performed for $N=500$ observations and $D=1000$ covariates.
\begin{itemize}
\item 4 Blocks: The covariates are divided into 4 blocks with correlation inside but not among the blocks. The covariates follow a multivariate normal distribution $\mathcal{N}_{D}(0,\Sigma)$, where $\Sigma_{i,j}=0.8*\ind{i=j\, \mathrm{ mod }\, 4}$
\item Toeplitz design: The correlation between two covariates is higher the closer their indices are. The covariates follow a multivariate Normal distribution $\mathcal{N}_{D}(0,\Sigma)$, with $\Sigma_{i,j}=0.99^{|i-j|}$
\item Toeplitz (grouped predictors): As Toeplitz design, but the indices of the informative covariates consist of 5 groups of 4 indices, each drawn uniformly in the interval $[100g-20, 100g+20]$ where g is the group number. Therefore the informative covariates exhibit a cluster structure.

\item 10 factors: Each covariate $X_d$ is generated as a linear combination of unknown latent factors \linebreak $X_d = \sum_{i=1}^{10} f_{d,i} \Phi_i + \nu_d \ \forall d=1,\ldots, D$
where the latent factors $\Phi_{i}$, and the noise $\nu_d$ follow a standard normal distribution.
The factor loading coefficients $f_{d,i}$ are fixed for any given realization of the dataset 
and are drawn beforehand from a $\mathcal{N}(0,1)$ distribution.
\item Correlated informative covariates, independent noise: The covariates follow a multivariate normal distribution $\mathcal{N}_{D}(0,\Sigma)$, with $\Sigma_{i,j}=0.9$ for all indices $i,j$ of informative
covariates and $\Sigma_{i,j}=0$ elsewhere.
\item Vitamin dataset: A gene expression dataset considered by \citet{meinshausen2010stability} (with $N=115$, $D=4088$.)
\end{itemize}

The only dataset that was not generated by the linear model is the Vitamin dataset. Following \citet{meinshausen2010stability} we choose 6 covariates that have high correlation with the target as signal covariates. We permute all other covariates across the samples to break its dependence with the target. The covariance structure of these noise covariates is kept as we use the same permutation for all covariates.

For all datasets except the Vitamin dataset we adjust the signal to noise ratio  $\Var(\langle X^{(n)},\beta \rangle) / \linebreak \Var(\eps^{(n)})=2$
for the first criterion and $8$ for the second and generate the vector $Y$ by the linear model given in Equation \eqref{linearModel}.

Finally, we denote an instance of the proposed method by $SFS(L,V)$ with $L,V$ as in Section \ref{sec:description}.

\subsubsection{Results with Precision@20 criterion}

{\bf Experimental protocol.} For the first criterion, we compare the number of truly informative covariates among the top $k=20$ selected, when ranked by selection frequency. This criterion is also known as precision@k in the information retrieval literature. We compare the proposed methods for each design with several choices of parameters, the usual stability selection and the Lasso as reference methods\footnote{ We used 
the R-package LARS \citep{hastie2012lars} as Lasso implementation.}. We consider the choices $SFS(L,1)$ for $L\in\set{2,4,8}$ to investigate the effect of the subsample size of observations, and $SFS(2,V)$ for $V\in \set{1,2,4,8}$ to investigate the effect of the amount of randomization of the base method by taking covariate subsets. Note that the $SFS(2,1)$ is the standard stability selection. 

We studied the performance of the different methods for a range of possible regularization parameters of the base method. More specifically, rather than comparing the methods for a grid of fixed
values of the regularization parameter $\lambda$ of Lasso, see Equation \eqref{eq:lasso}, for each realization
of the data, we used the values $\lambda_{q_L}$ such that exactly $q_L  \in \{1,\dots,100\}$ covariates are
selected by the base method. (In case several such values exist, the largest one is picked.)
This approach seemed more fair and in line with the analysis of Section \ref{se:anal}, in the sense
that the number of covariates selected by the baseline, $q_L =\e{\abs{S^{\mathrm{base}}_{L,V}}}$, is kept constant across choices of $(L,V)$. For each value of $q_L $, we report the average precision@k, i.e., how many
informative covariates are selected amongst the top $k=20$ covariates ranked by their
selection frequency. For standalone Lasso, we ranked the covariates by their estimated regression coefficients. The performance reported in Figure \ref{fig:resultsObservations} and \ref{fig:resultsCovariates} is an average over 100 repetitions of each experiment.

{\bf Results.} A first finding from Figure \ref{fig:resultsObservations} is that 
standard stability selection ($SFS(2,1)$) does not systematically outperform standalone Lasso, which appears to contradict the results of \citet{meinshausen2010stability}.
A first reason for this is that here are some substantial differences between our evaluation criterion
and theirs; this is discussed in more detail below.
Perhaps more importantly, we evaluate the performance of standalone Lasso differently.
\citet{meinshausen2010stability} consider a single
ranking of covariates by their selection order along the ``Lasso path'',
while we rank covariates by the
magnitude of their estimated Lasso coefficients (for each value of $q_L $ taken
separately).

We find that using our ranking allows standalone Lasso to recover informative covariates more successfully when it first selects a total number of covariates higher than the true actual number of informative ones,
then only keeps the ones with largest coefficient magnitude.
In other words, relevant covariates might at first not be included in the selection set along 
the Lasso path, but once they are, they tend to get estimated coefficients larger than those of noisy covariates.
As  a consequence, in the regime where $q_L $ is markedly larger than the true number of informative covariates, the success of stability
selection over standalone Lasso is not systematic and appears to depend on the
setting.

\begin{figure*}
  \includegraphics[height=0.29\textwidth]{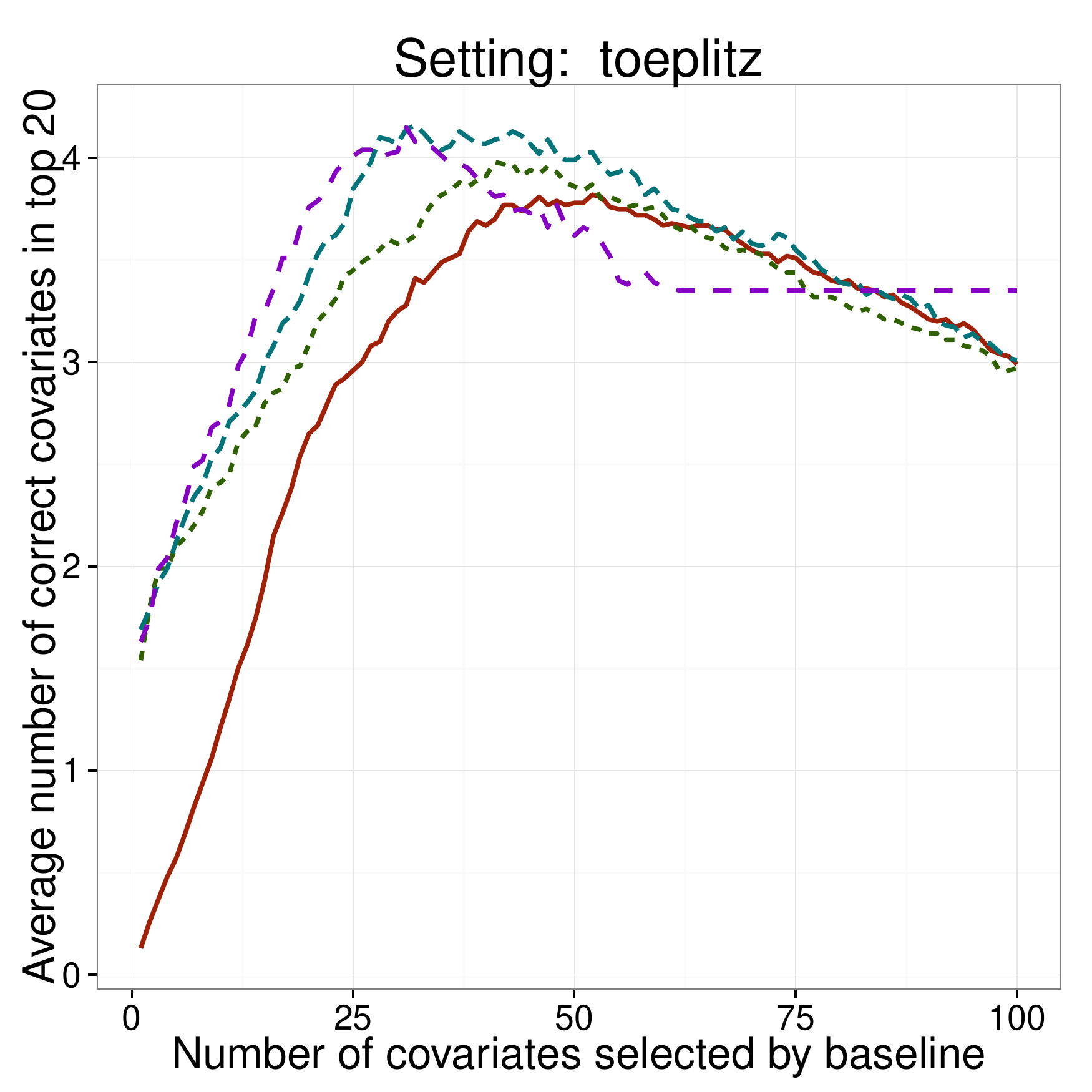}\nolinebreak
  \includegraphics[height=0.29\textwidth]{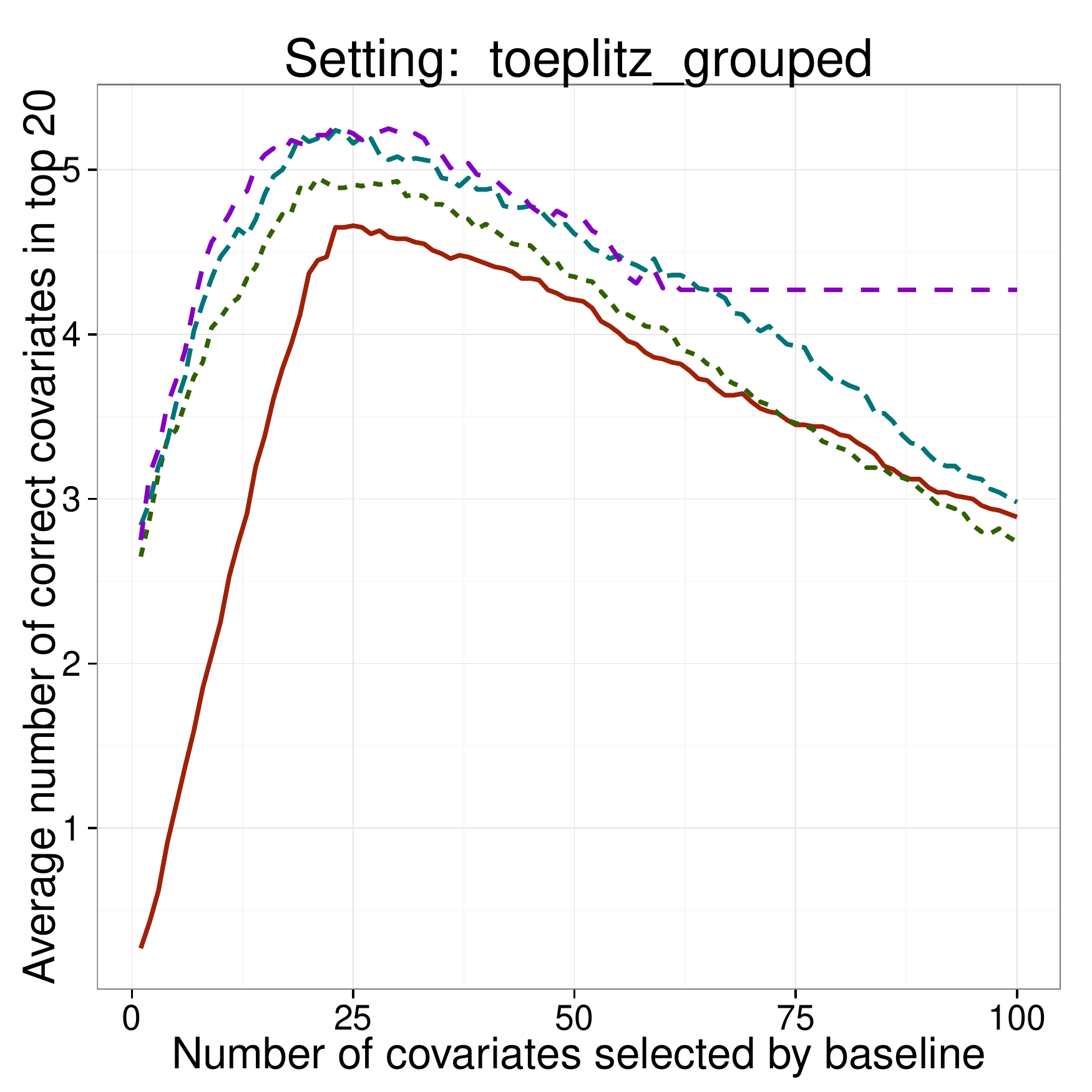}\nolinebreak
  \includegraphics[height=0.29\textwidth]{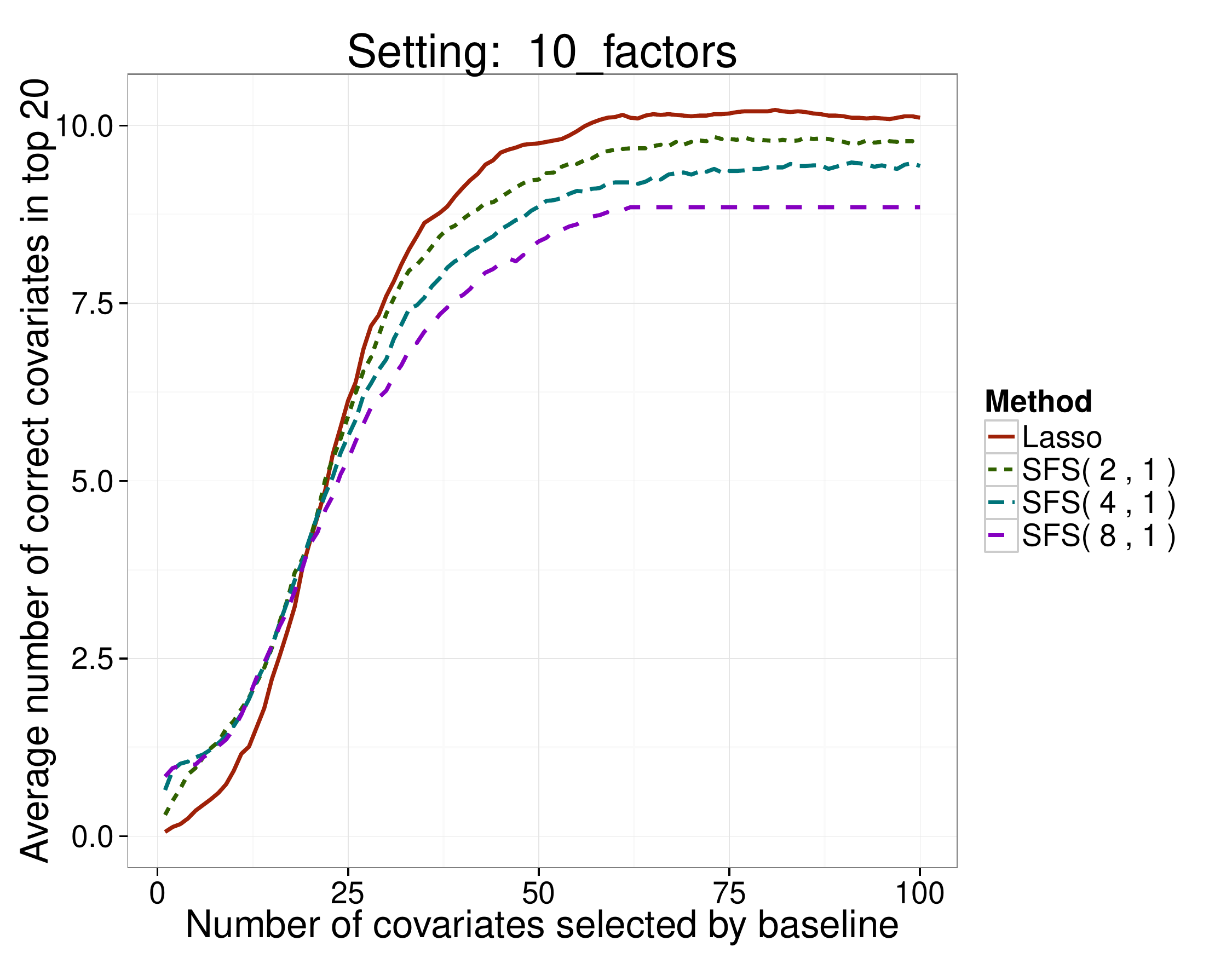}\\
  \includegraphics[height=0.29\textwidth]{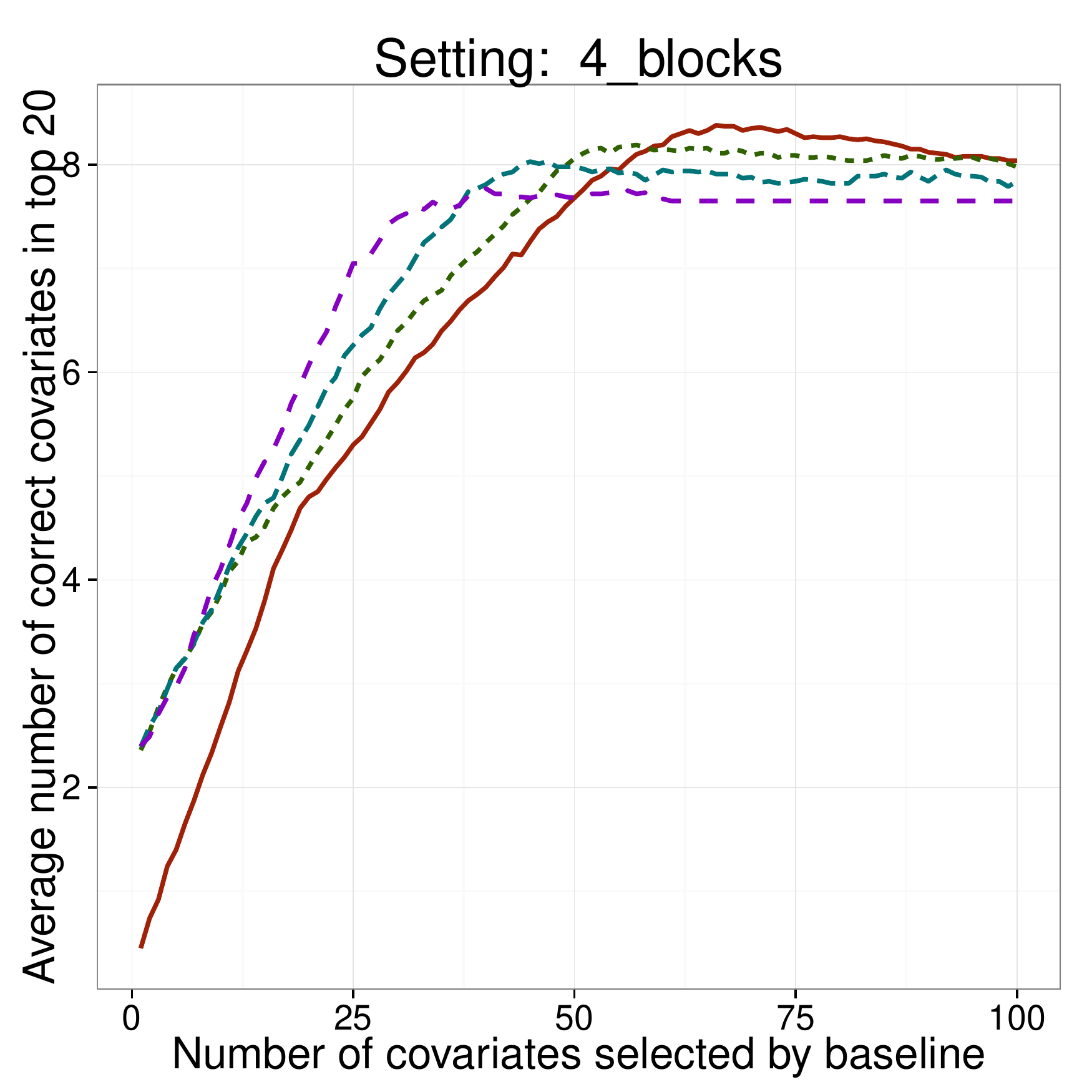}\nolinebreak
  \includegraphics[height=0.29\textwidth]{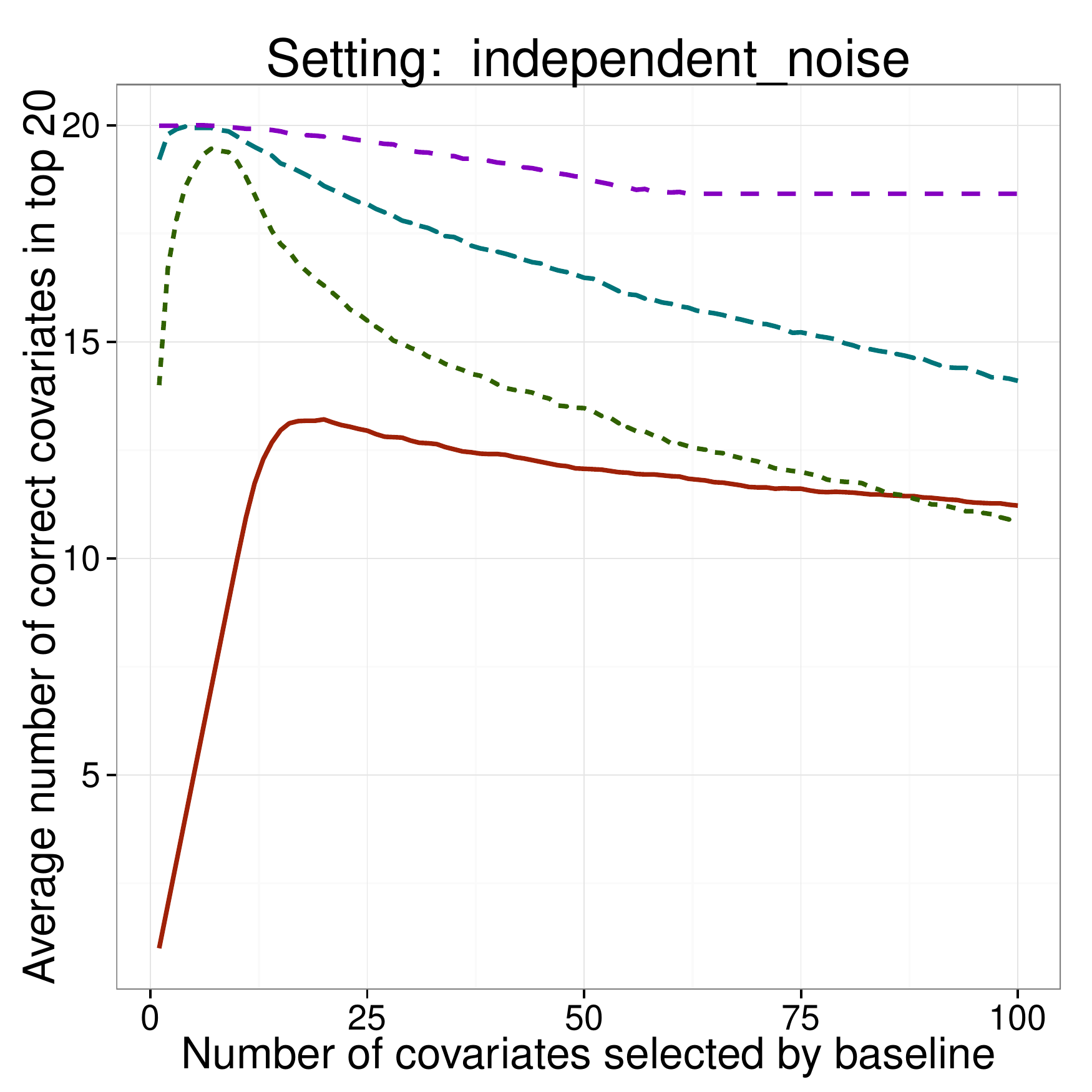}\nolinebreak
  \hspace{0.24cm}\includegraphics[height=0.29\textwidth,width=4.94cm]{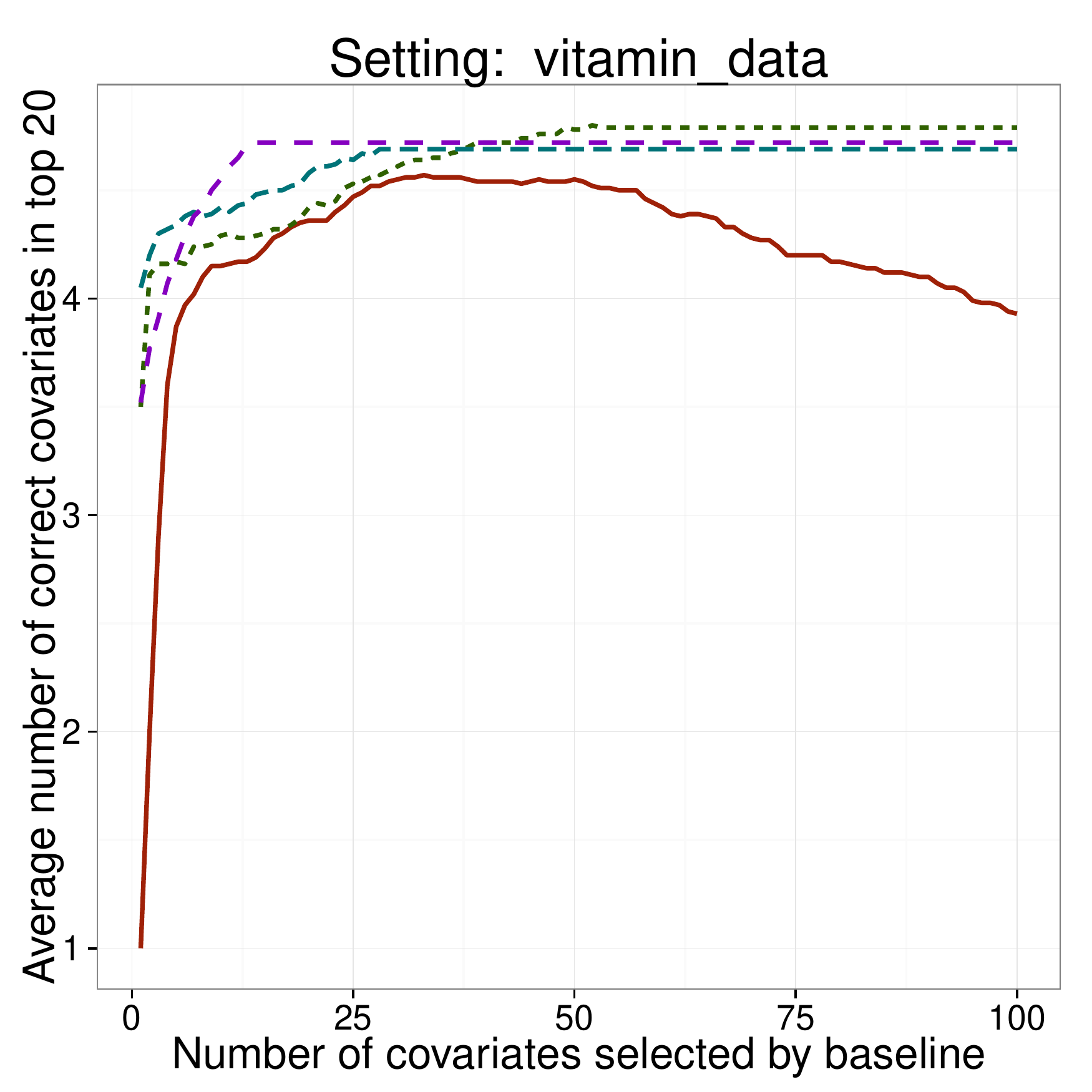}
  \caption{Comparison of plain Lasso and stability selection with varying numbers of disjoint observation subsamples (corresponding to line types). We plot the average number of informative covariates among the top 20 scored (for the Vitamin dataset: top 6), depending on the number of covariates selected by the base method, figures correspond to designs. Each dataset contains 20 signal and 980 noise covariates, except for the Vitamin dataset, which contains 6 signal and 4082 noise covariates.}
  \label{fig:resultsObservations}
\end{figure*}

\begin{figure*}
  \includegraphics[height=0.29\textwidth]{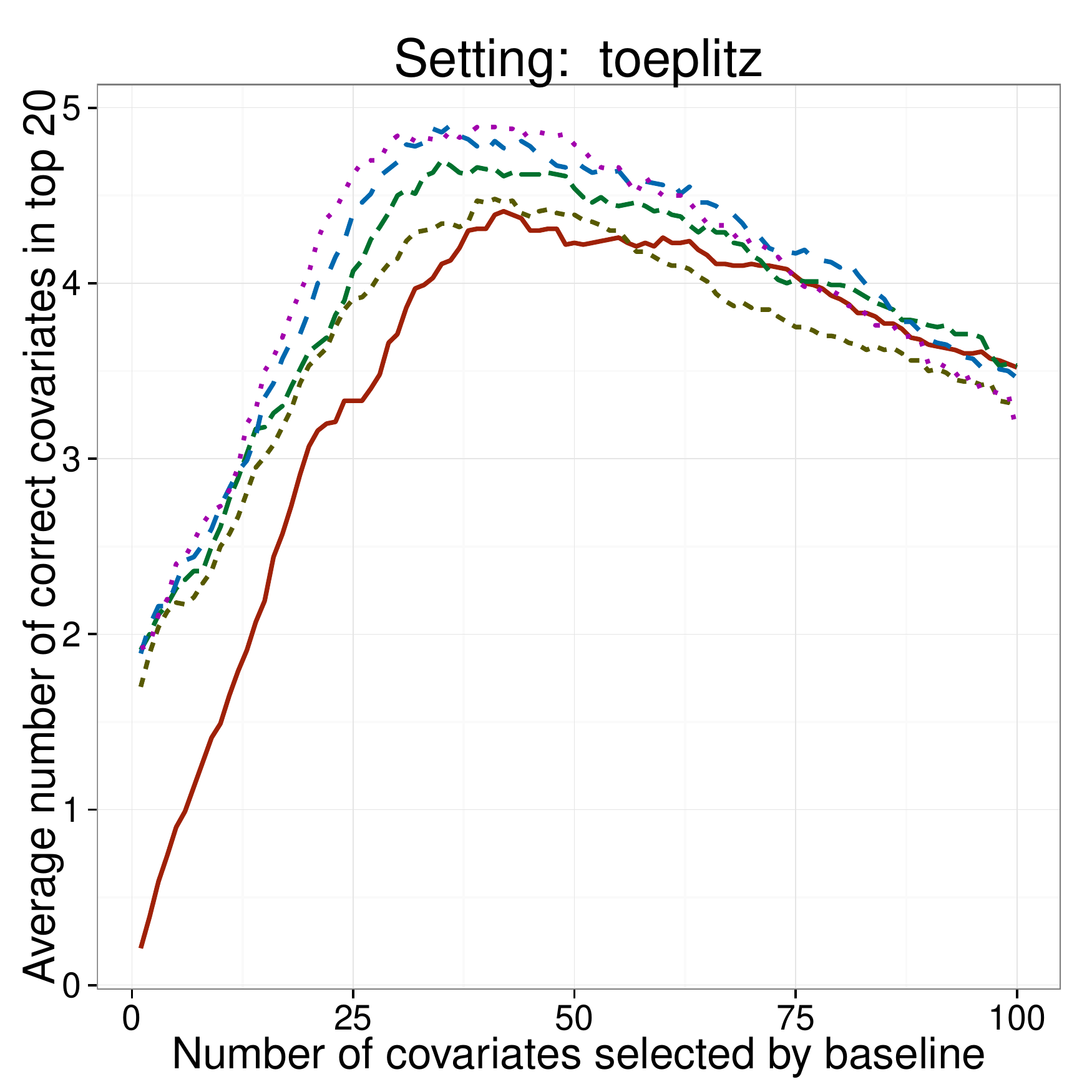}\nolinebreak
  \includegraphics[height=0.29\textwidth]{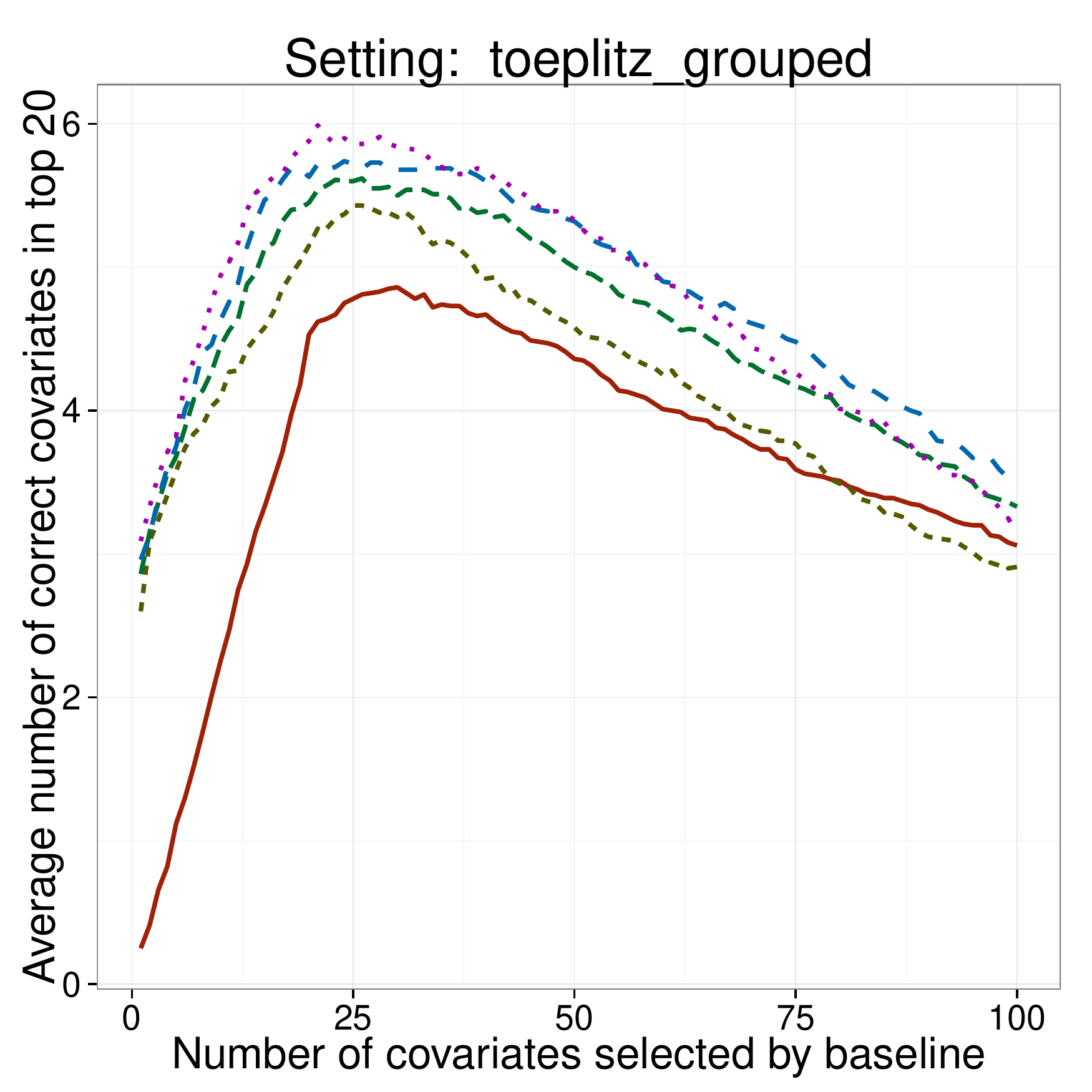}\nolinebreak
  \includegraphics[height=0.29\textwidth]{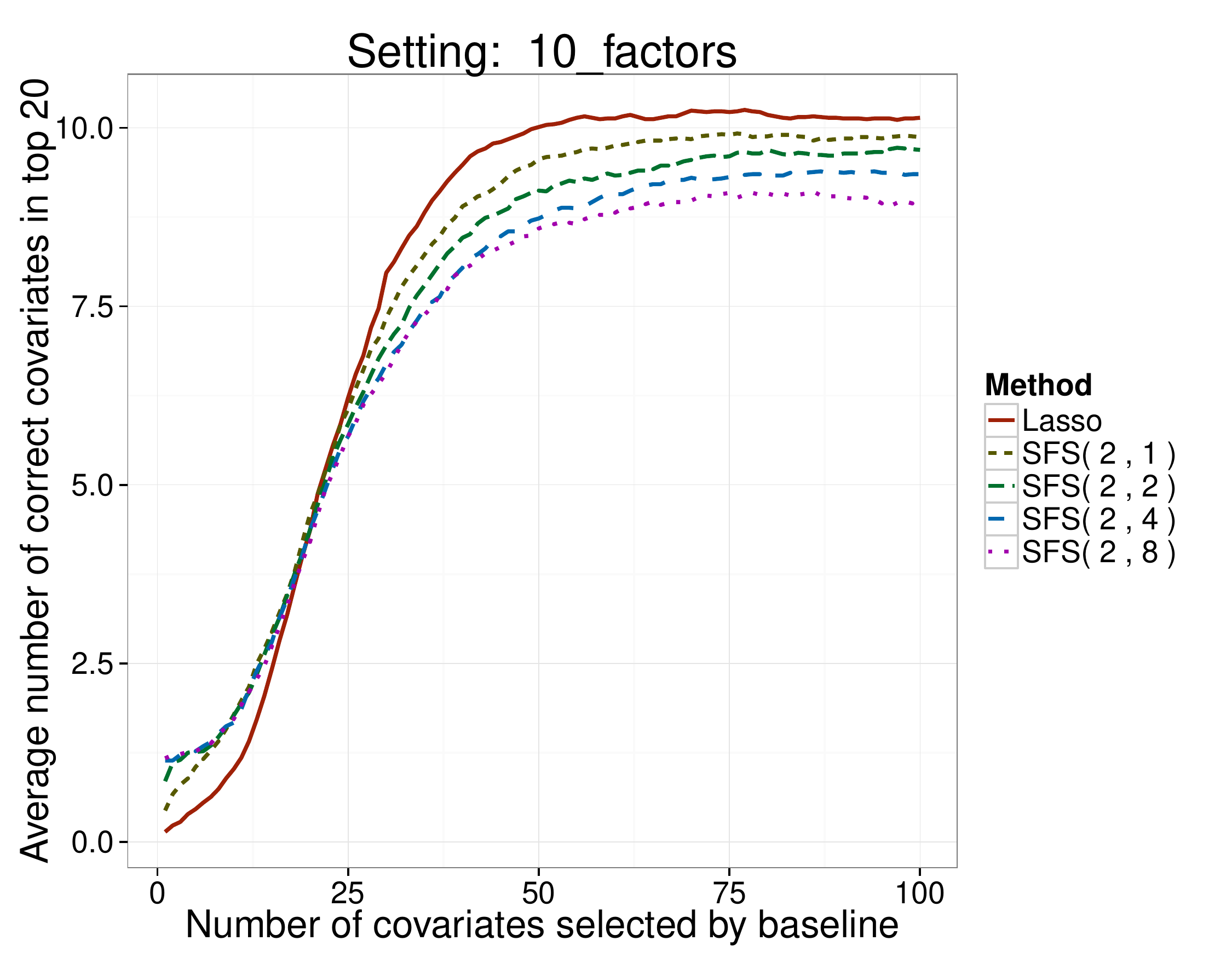}\\
  \includegraphics[height=0.29\textwidth]{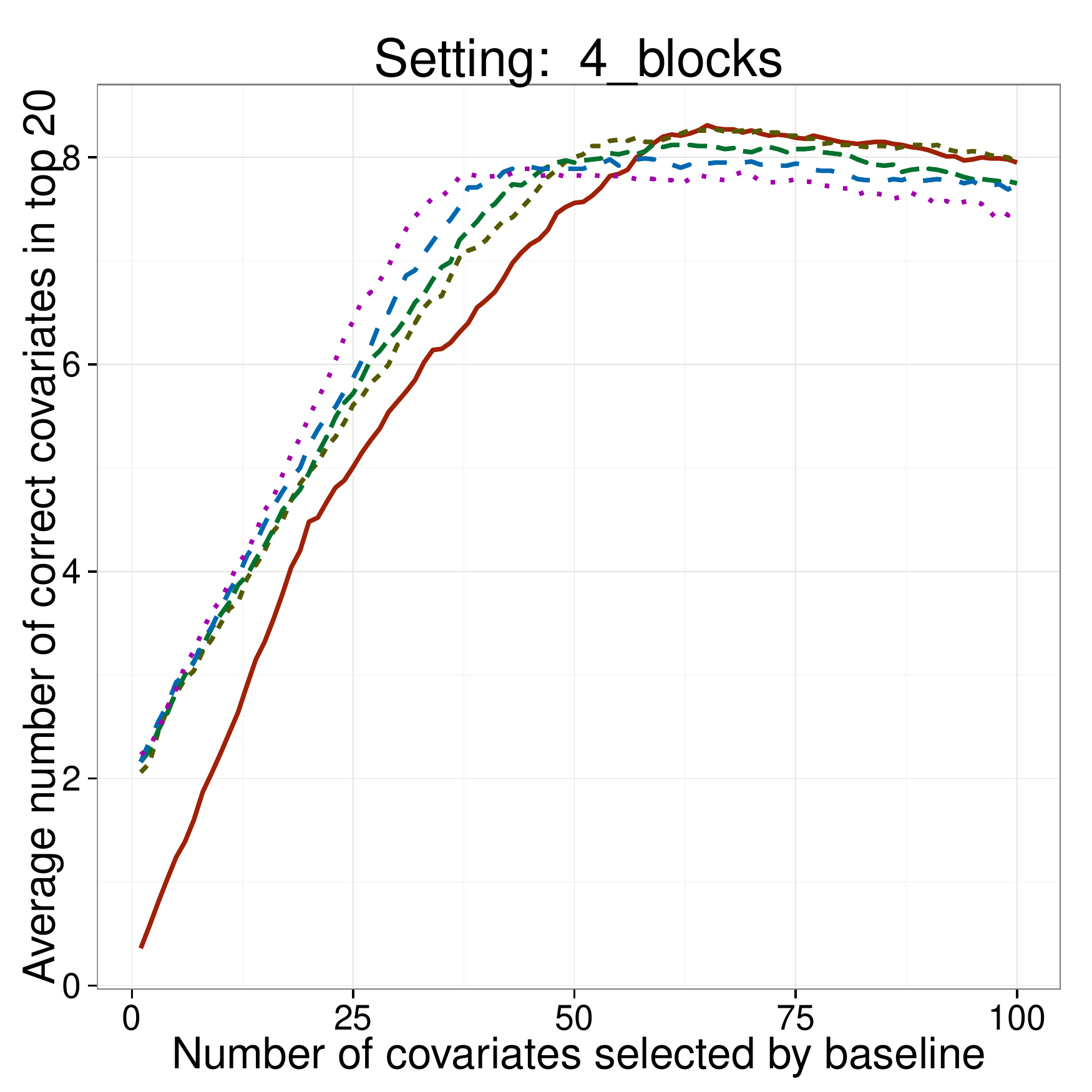}\nolinebreak
  \includegraphics[height=0.29\textwidth]{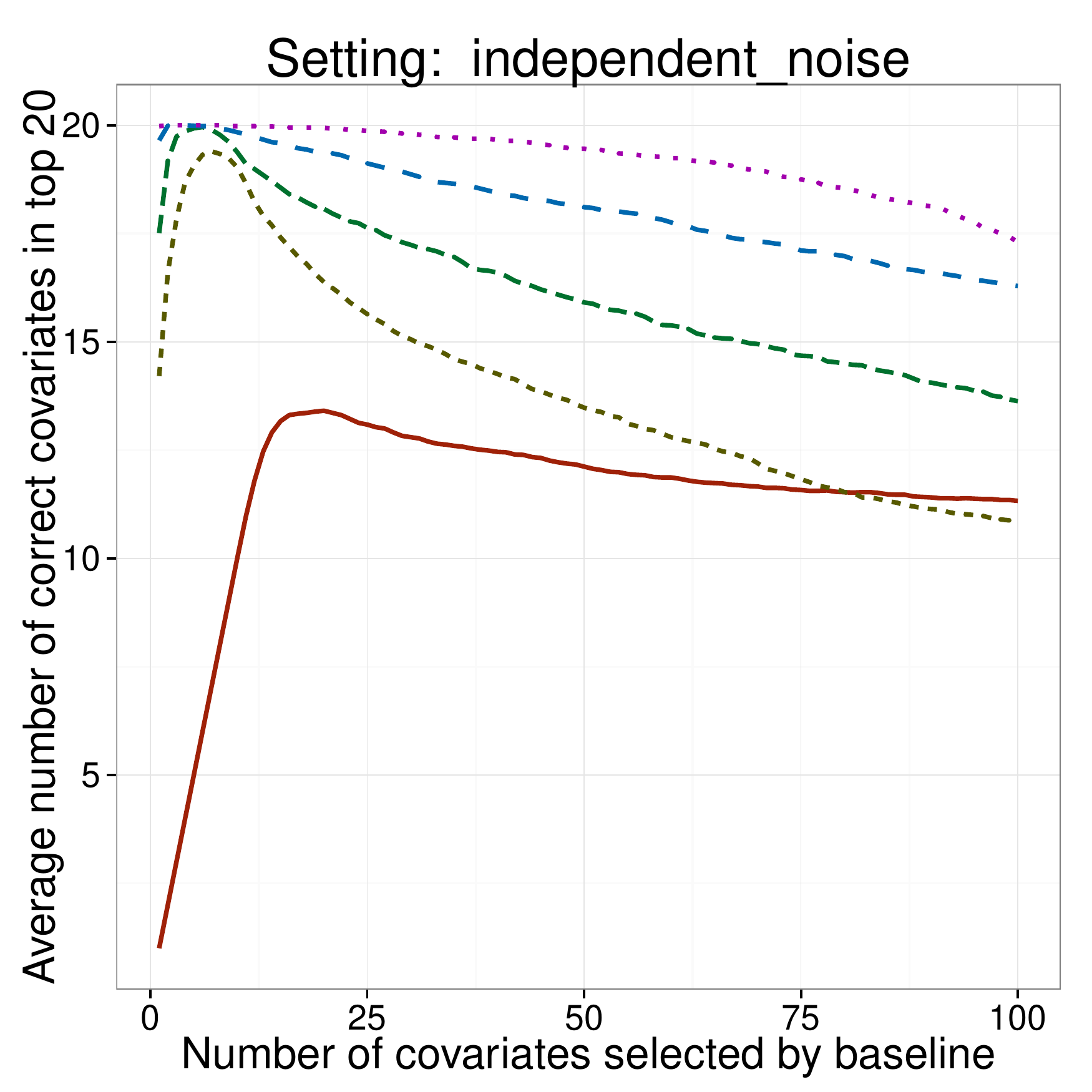}\nolinebreak
  \hspace{0.24cm}\includegraphics[height=0.29\textwidth,width=4.94cm]{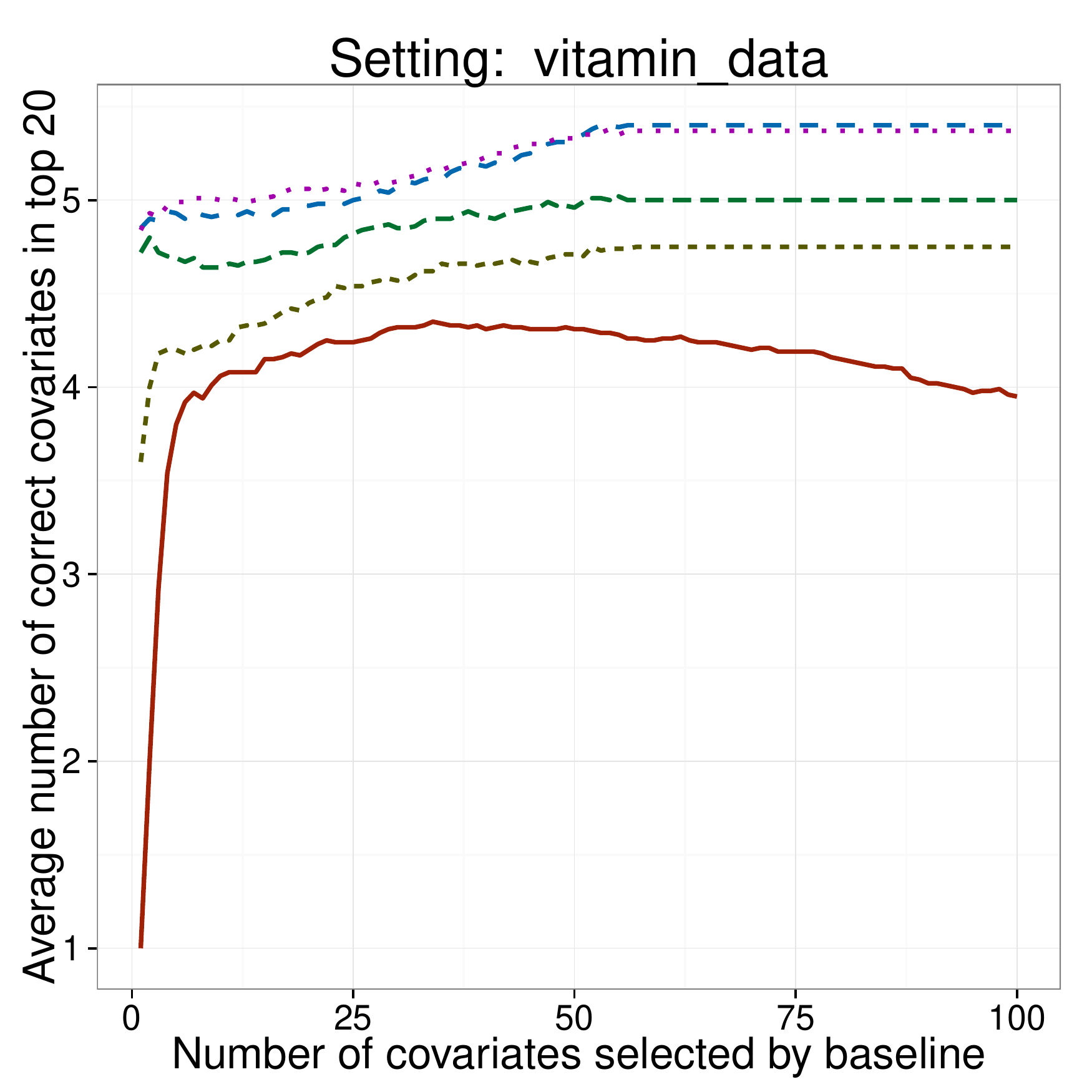}
  \caption{Like Figure \ref{fig:resultsObservations}, except that we fix the number of disjoint observation subsamples and vary the number of disjoint covariate subsets.
    \label{fig:resultsCovariates} 
  }
\end{figure*}

In additional experimental results presented in the supplemental material, we investigated the possibility of taking advantage of this effect by modifying the base method: We allowed Lasso to select $40$ covariates, ranked them by the magnitude of their coefficients and kept a fixed number $q_L$ (between $1$ and $40$) of covariates
with largest coefficients. This is a version of the thresholded Lasso, that has been shown to have favourable asymptotic properties compared to the regular Lasso, see for instance \cite{lounici2008sup} and \\ \cite{meinshausen2009lasso}. However, we found that the precision@20 results of applying SFS on top of this base method are remarkably stable and unaffected by the value of $q_L$ in this setting\,.
Thus, thresholded Lasso and SFS does not improve significantly over regular Lasso and SFS.

Second, we observed in Figure \ref{fig:resultsObservations} and \ref{fig:resultsCovariates} that in most cases if for a fixed value of $q_L $ standard stability selection $SFS(2,1)$ 
improves on plain Lasso, then the extended methods we propose generally lead to a further improvement.
A qualitative observation is that stability selection appears to be
more successful for situations with limited or short-range dependence (Toeplitz
design, Toeplitz grouped, independent noise, realistic datasets with presumably weak dependence)
than with systematic dependence structure (factorial design, block design).
Naturally, and as expected from the analysis and discussion in Section \ref{se:anal}, 
eventually the performance deteriorates again if the size of the observation subsamples $N/L$ 
becomes too small to allow a reasonable estimation, or if the number of 
covariates $D/V$ used in the randomized base procedure is too low 
(experiments not shown). Overall, these results suggest to use extended stability
selection, both for reasons of potential increase in performance, and
of efficient scalability via possible parallelization over standard
stability selection, whenever the latter itself improves on the base method.

{\bf Discussion of the precision@20 criterion.}
In their original work \citet{meinshausen2010stability}
reported two evaluation measures: the first was the probability that the top $\ell$
covariates by selection frequency are all relevant (for some fixed beforehand $\ell$).
The second was the number of noise 
variables included in the top-$\ell$ selection, where $\ell$ is chosen 
(separately in each data realization) so that
the selection contains at least a fixed proportion such as 20\% of the true variables. 
In experiments using these criteria, 
we found that the first criterion was only poorly informative, 
in the sense
that it was often very close to either 0 or 1, not providing a very clear contrast between the
methods; and that the second criterion was 
in many cases subject to a very large 
expectation and variance, so that we were also wary about its relevance.
For this reason, we chose precision@20 as the performance measure; this
criterion has the advantage of being stable, comparable across settings, in plausible 
relation to intended applications, and is standard in information retrieval.
Additionally, for stability selection \citet{meinshausen2010stability} 
only reported results for a specific value of $q_L $, namely $q_L \approx \sqrt{0.8 D}$,
and we preferred to report performance over a range of values of $q_L $.
Overall we believe our protocol presents a more complete and fairer overall picture.

\subsubsection{Results with theoretical control of false positives}

In this section we first illustrate and validate the bound given by Corollary \ref{co:sampleSplitting}.

We now perform experiments similar to those described above, to investigate the effect of observation subsampling while using the value of $\tau$ dictated by Corollary \ref{co:sampleSplitting} in order to achieve less than one expected false positive (i.e. the values represented in Figure~\ref{fig:tauQ}).  We determine the average number of false positives for $q_L\in \{1,..,100\}$. After comparing the empirical number of false positives to the theoretical bound, we also count the number of true positives to determine the power of the method. As we noticed that for a signal to noise ratio of $2$ hardly any covariates are selected at all, we used a signal to noise ratio of $8$. (In absence of a false positive error bound
for standalone Lasso, the latter is not included in these experiments.)
The results are given in Figures \ref{fig:resultsFp} and \ref{fig:resultsTp}. The main findings concerning the comparison of standard stability selection and extended stability selection are the following:

\begin{itemize}
\item The bound for the number of false positives holds in all experiments except for factorial and grouped Toeplitz design when $q_L$ is small;
it seems to indicate that the
symmetry assumptions of Corollary 1 are significantly violated in these situations.
  \item Extended stability selection selects fewer false positives in all experiments.
  \item Extended stability selection selects more true positives in all experiments except for the Vitamin dataset.
  \item The largest number of true positives is often achieved for some $q_L$ that is larger than the number of relevant covariates and often lies outside of the regime where the bound for standard stability selection can be used.
\end{itemize}
We observe in Figure \ref{fig:resultsFp} that the FP bound for extended stability selection seems to be loose.
There is probably room for improvement in the theoretical bounds, for instance using ideas from \cite{shah2013variable} to improve
on Markov's inequality under additional assumptions on the distribution of the frequency counts, though this
is outside of the intended scope of this paper.

\begin{figure*}
  \includegraphics[height=0.29\textwidth]{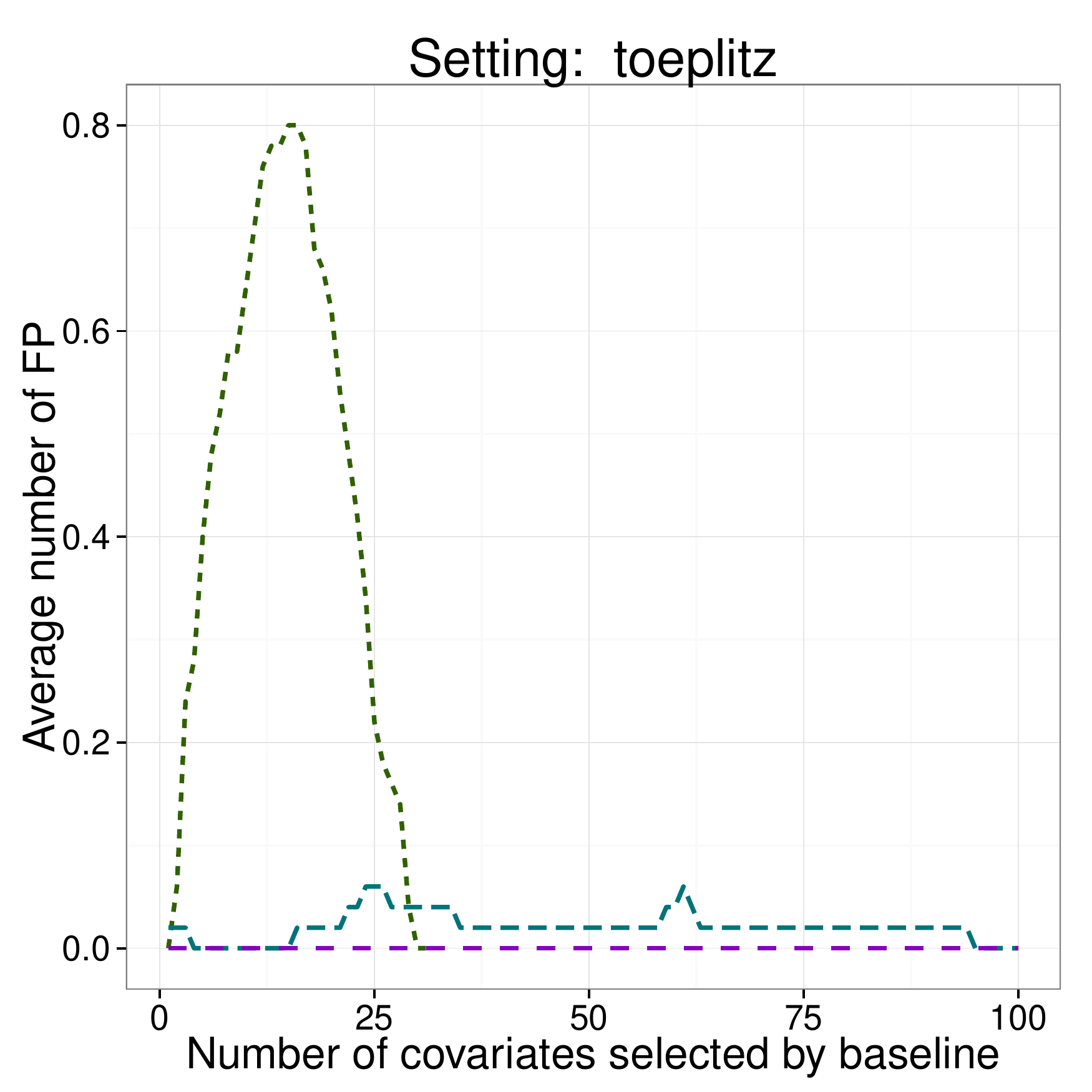}\nolinebreak
  \includegraphics[height=0.29\textwidth]{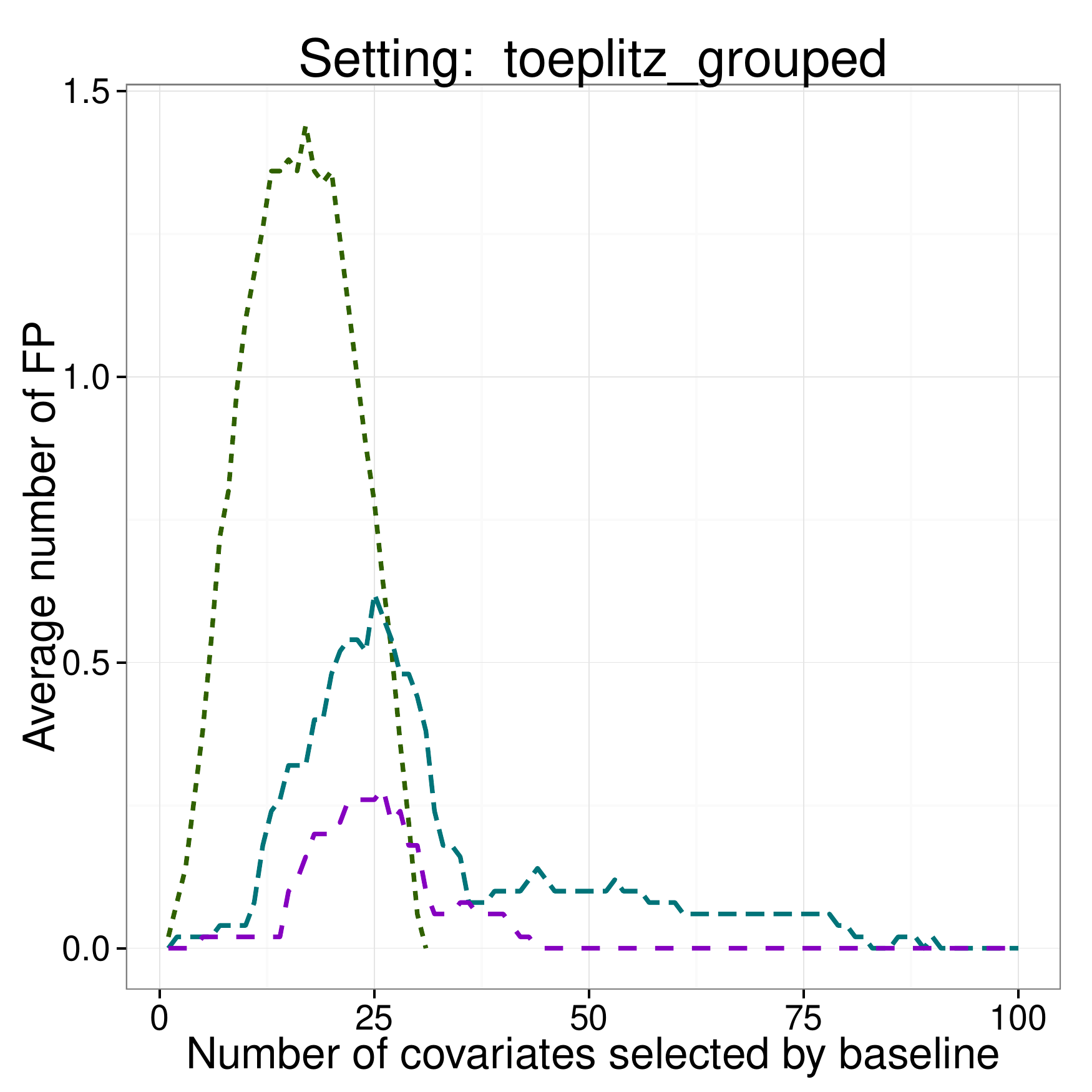}\nolinebreak
  \includegraphics[height=0.29\textwidth]{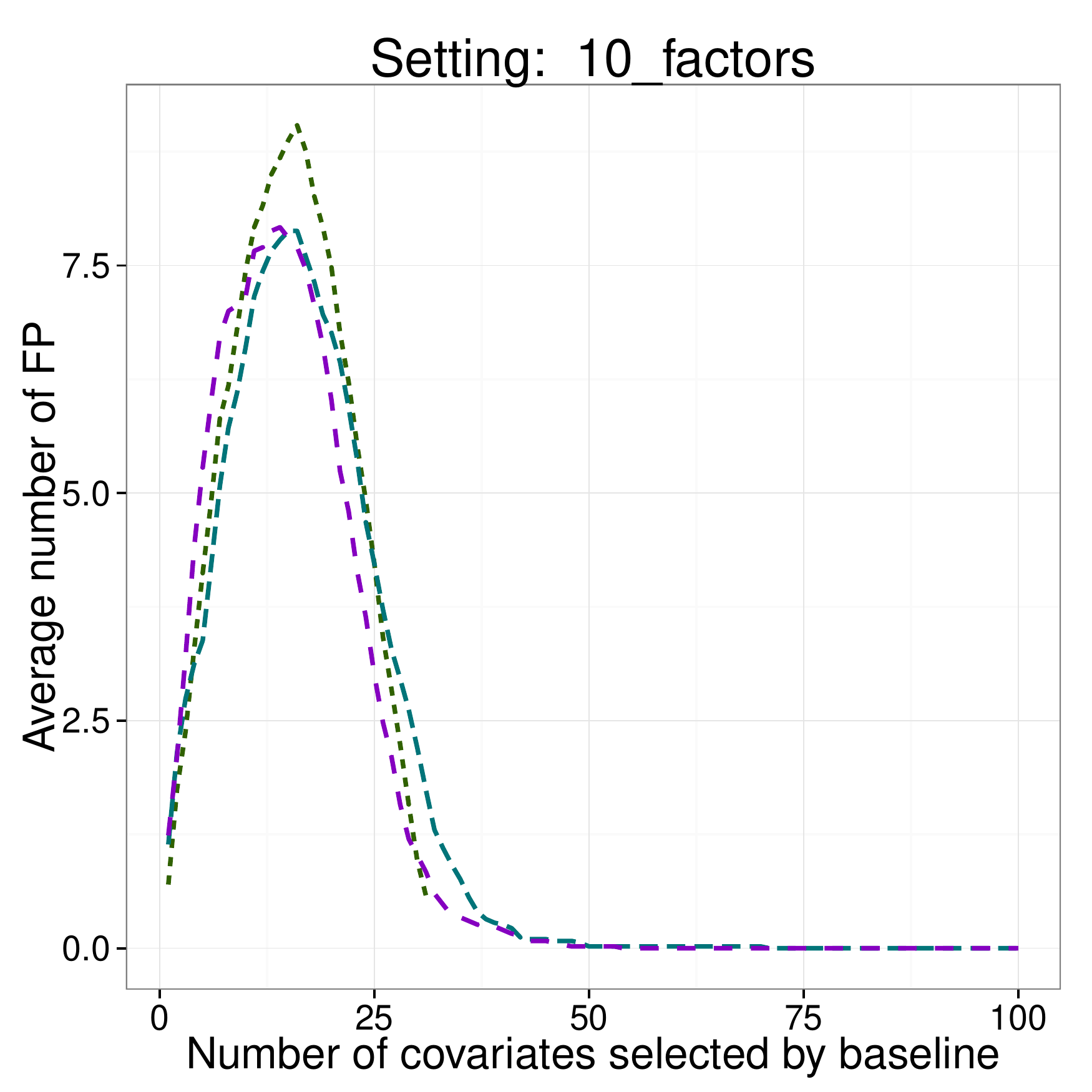}
  \includegraphics[height=0.29\textwidth]{legendFp}\\
  \includegraphics[height=0.29\textwidth]{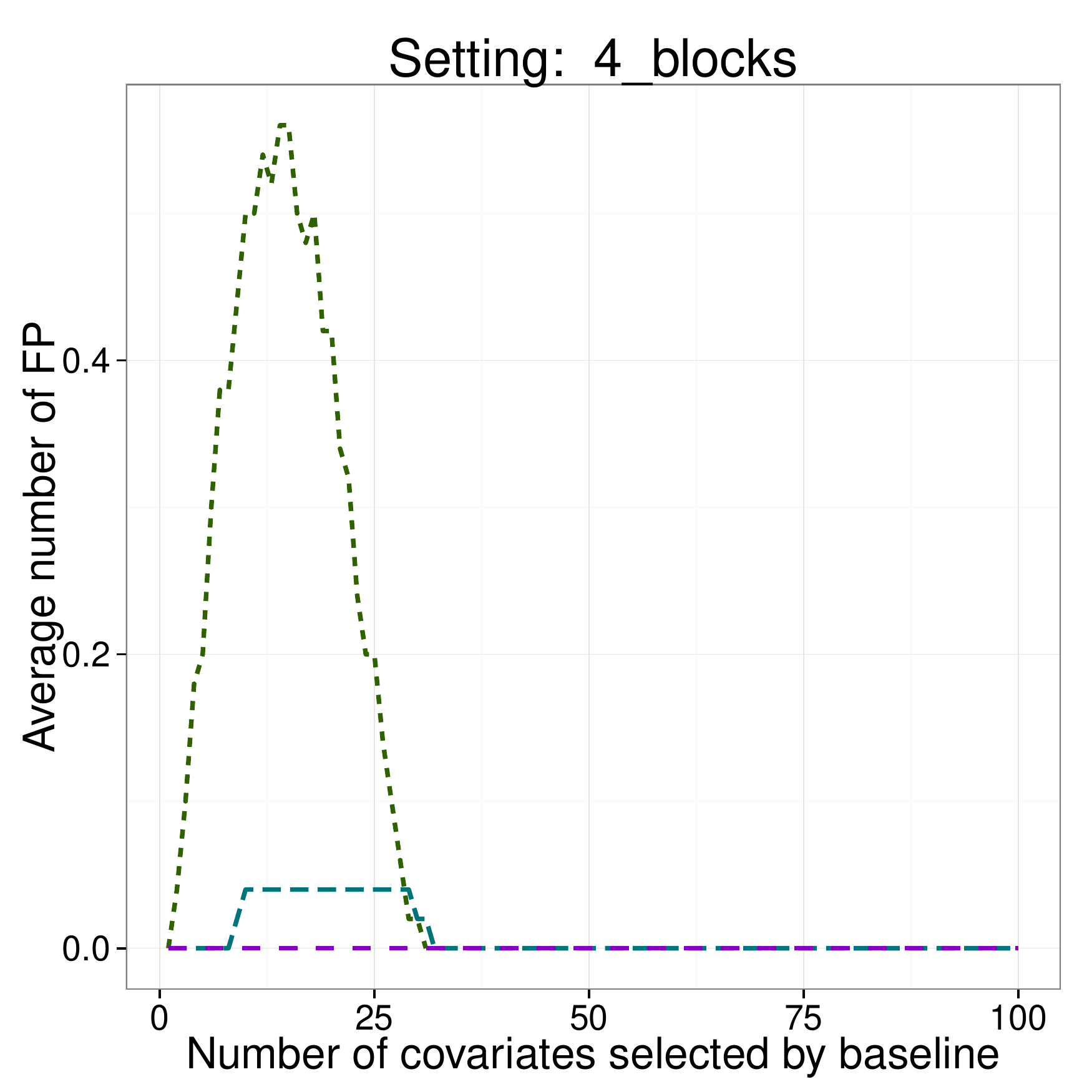}\nolinebreak
  \includegraphics[height=0.29\textwidth]{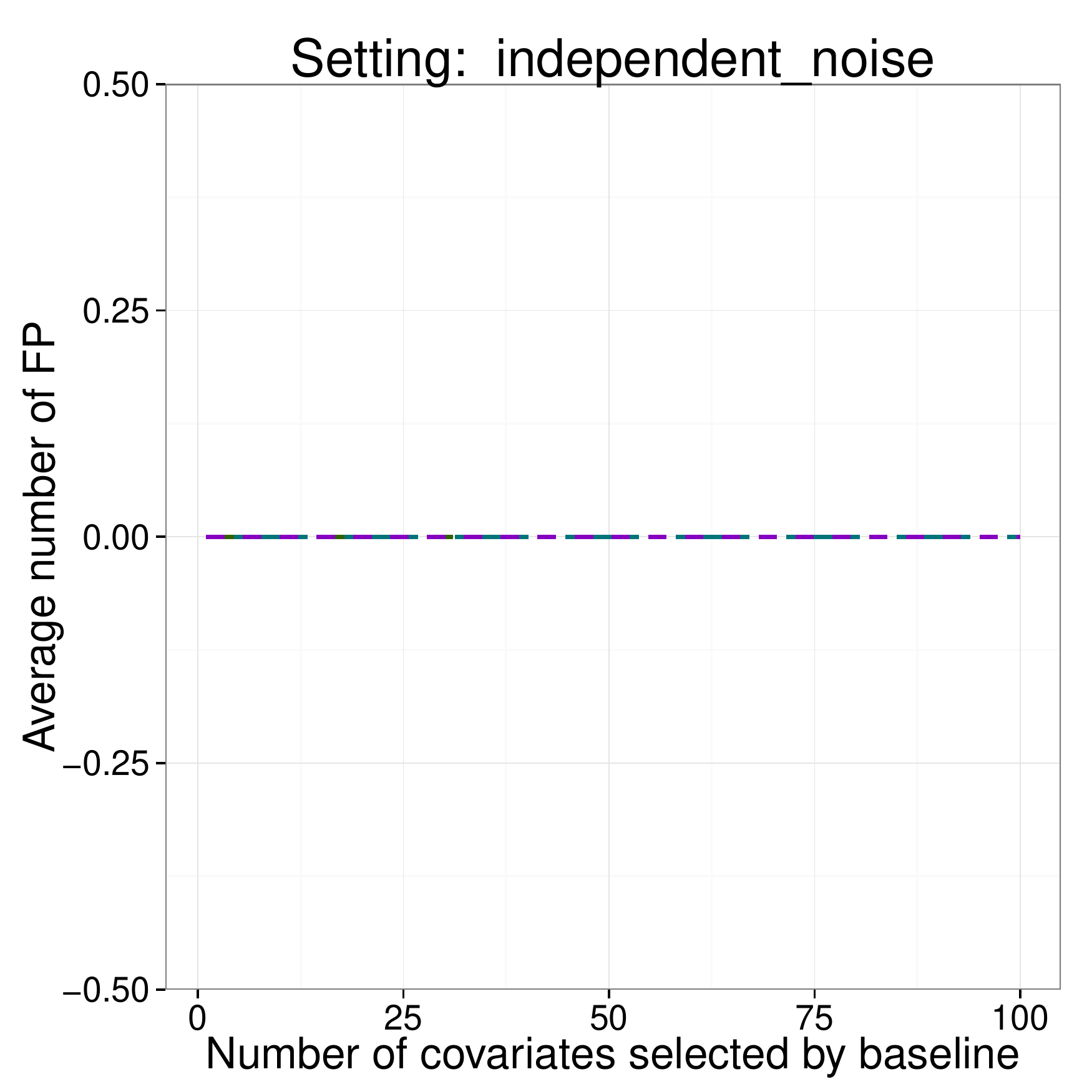}\nolinebreak
  \includegraphics[height=0.29\textwidth]{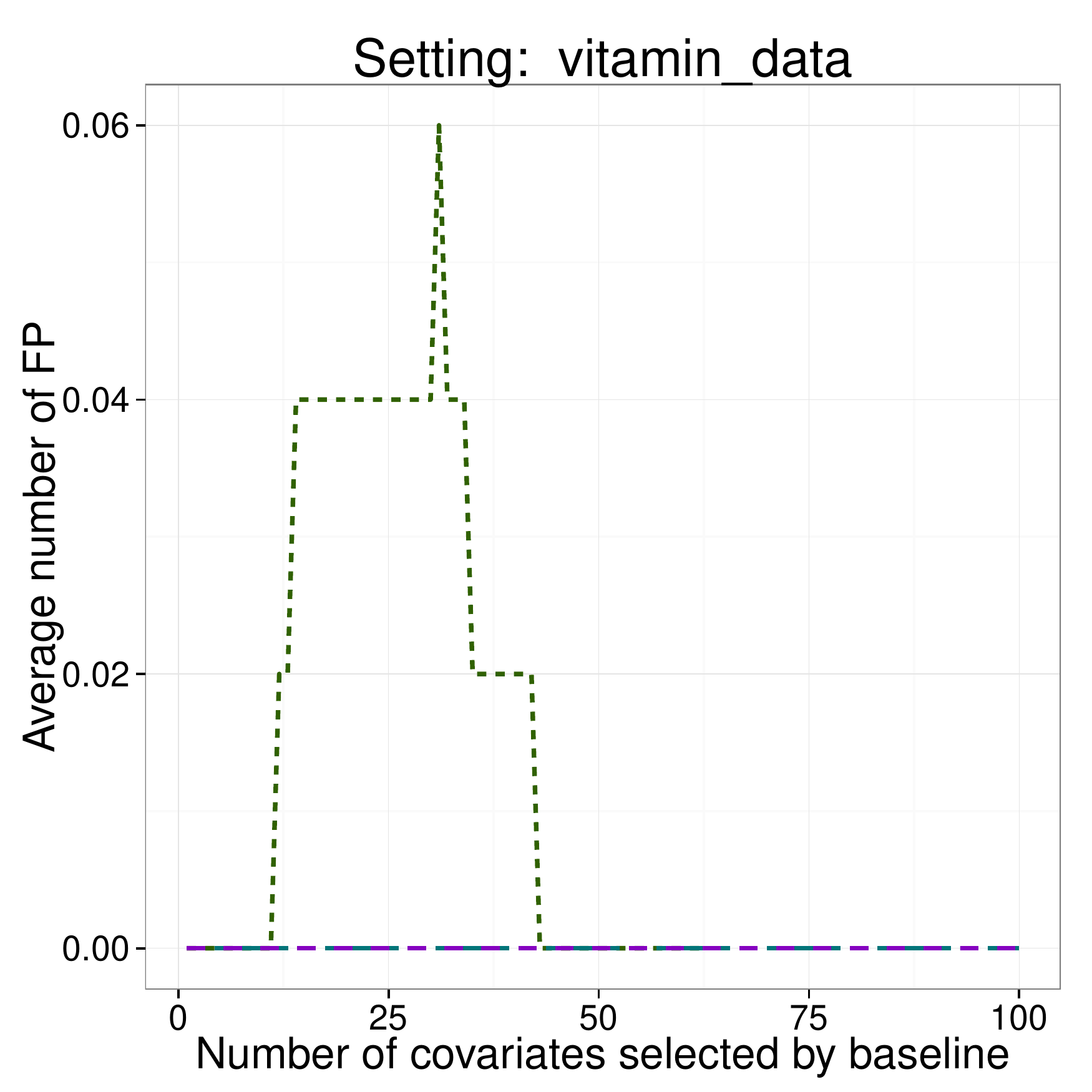}\nolinebreak
  \caption{Comparison of stability selection and its extensions for various numbers of observation subsamples (corresponding to line types), we choose the largest $\tau$ such that in Corollary \ref{co:sampleSplitting} the expected number of false positives is less than one. We plot the average number of false positives, depending on the number of covariates selected by the base method, figures correspond to designs. Each dataset contains 20 signal and 980 noise covariates, except for the Vitamin dataset, which contains 6 signal and 4082 noise covariates.} 
  \label{fig:resultsFp}
\end{figure*}

\begin{figure*}
  \includegraphics[height=0.29\textwidth]{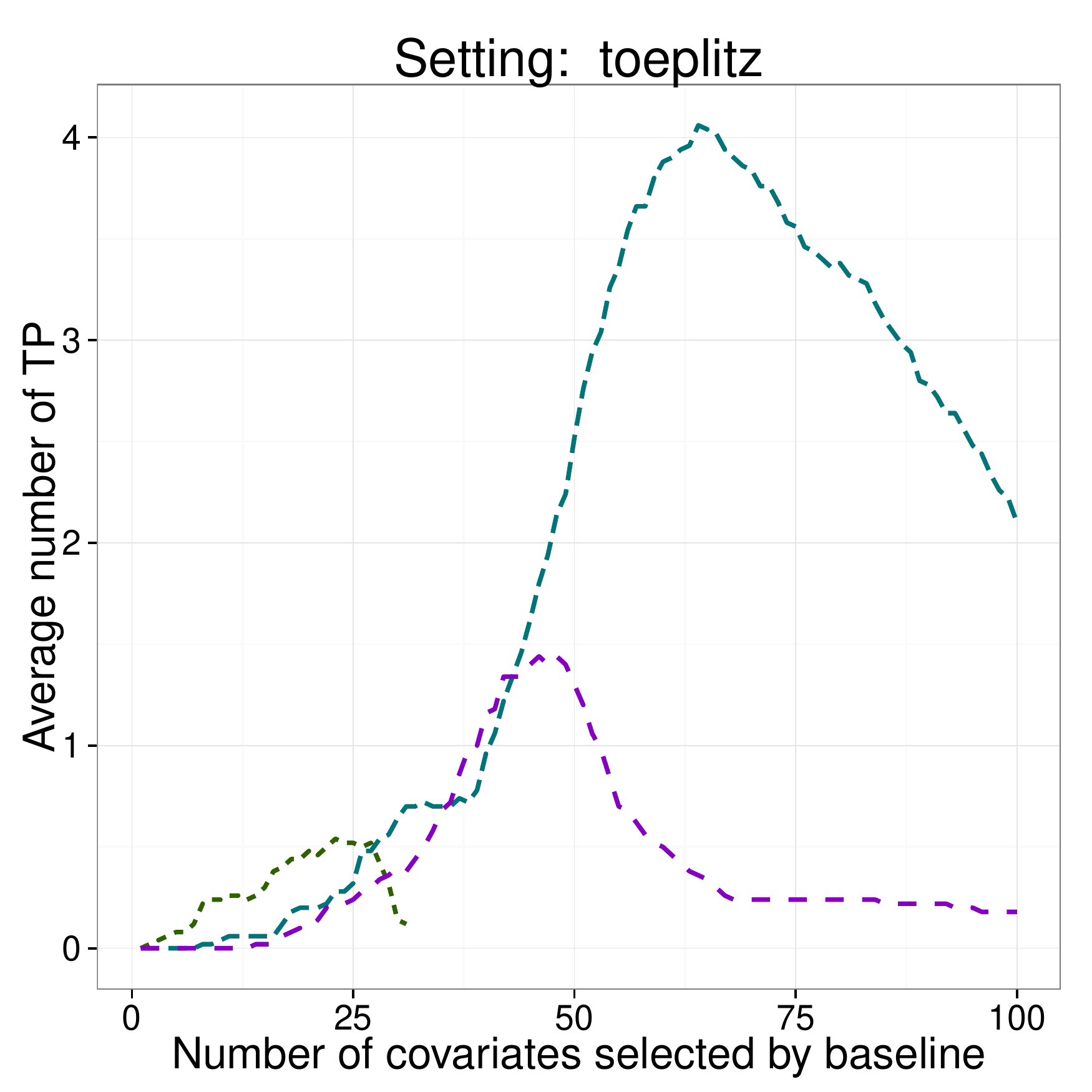}\nolinebreak
  \includegraphics[height=0.29\textwidth]{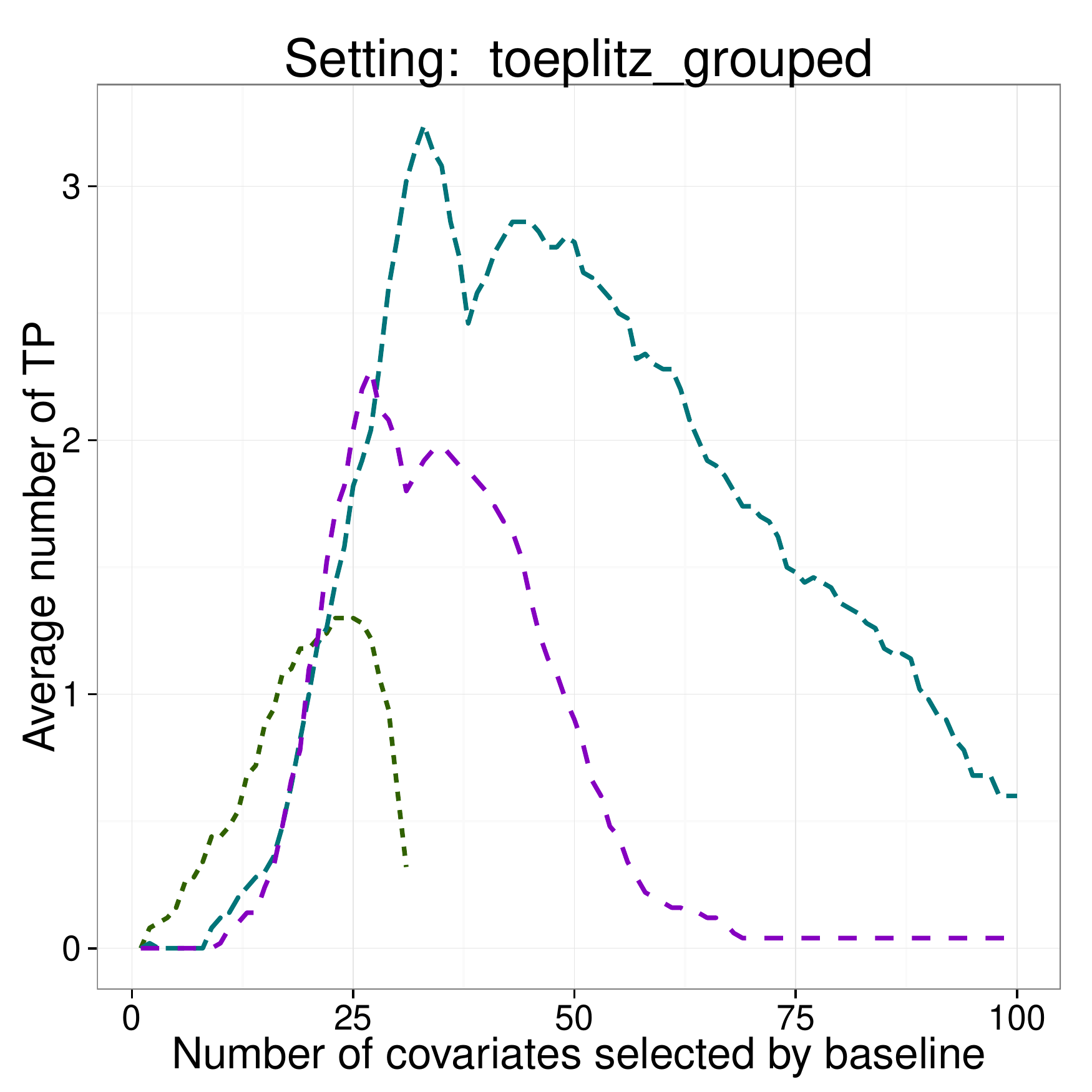}\nolinebreak
  \includegraphics[height=0.29\textwidth]{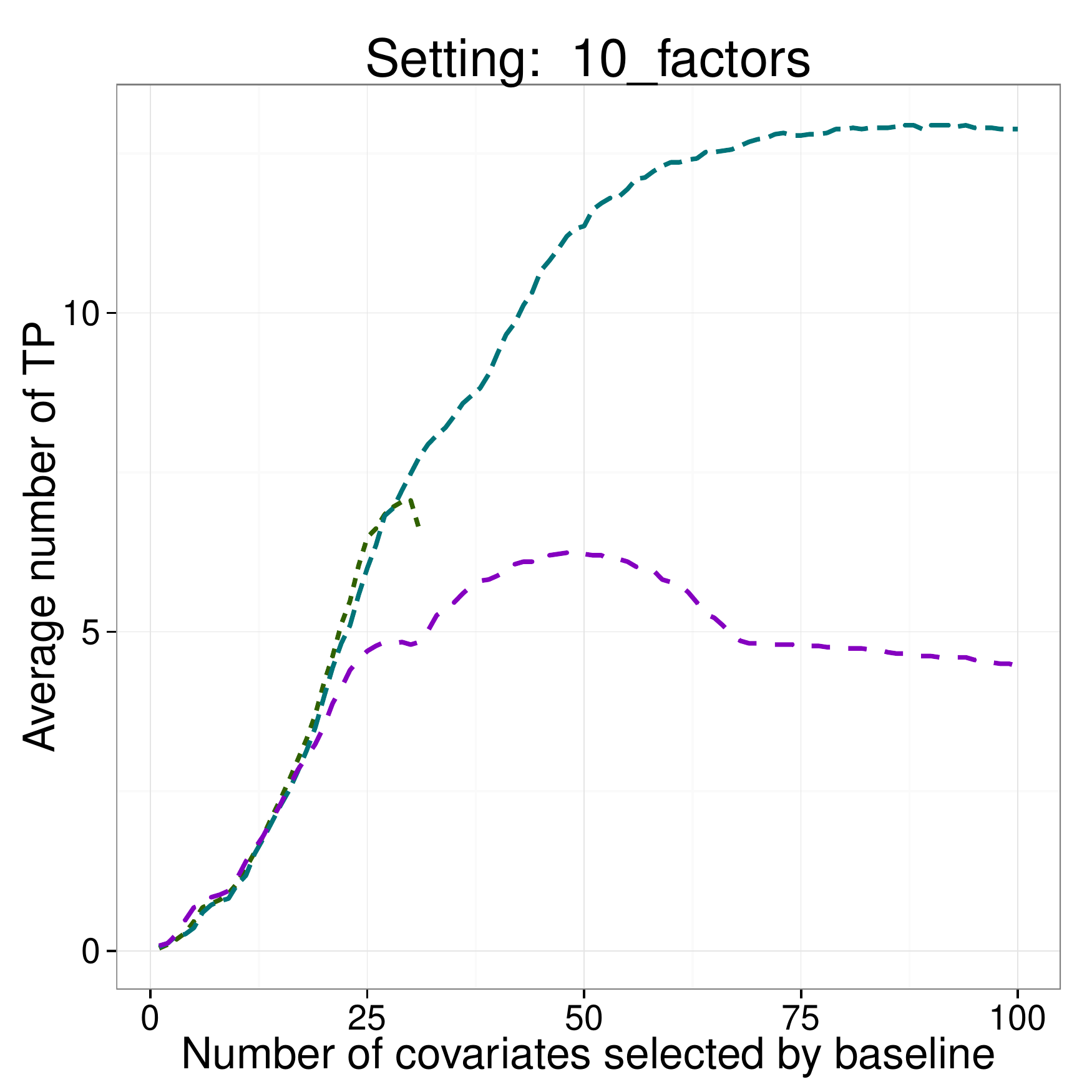}
  \includegraphics[height=0.29\textwidth]{legendFp}\\
  \includegraphics[height=0.29\textwidth]{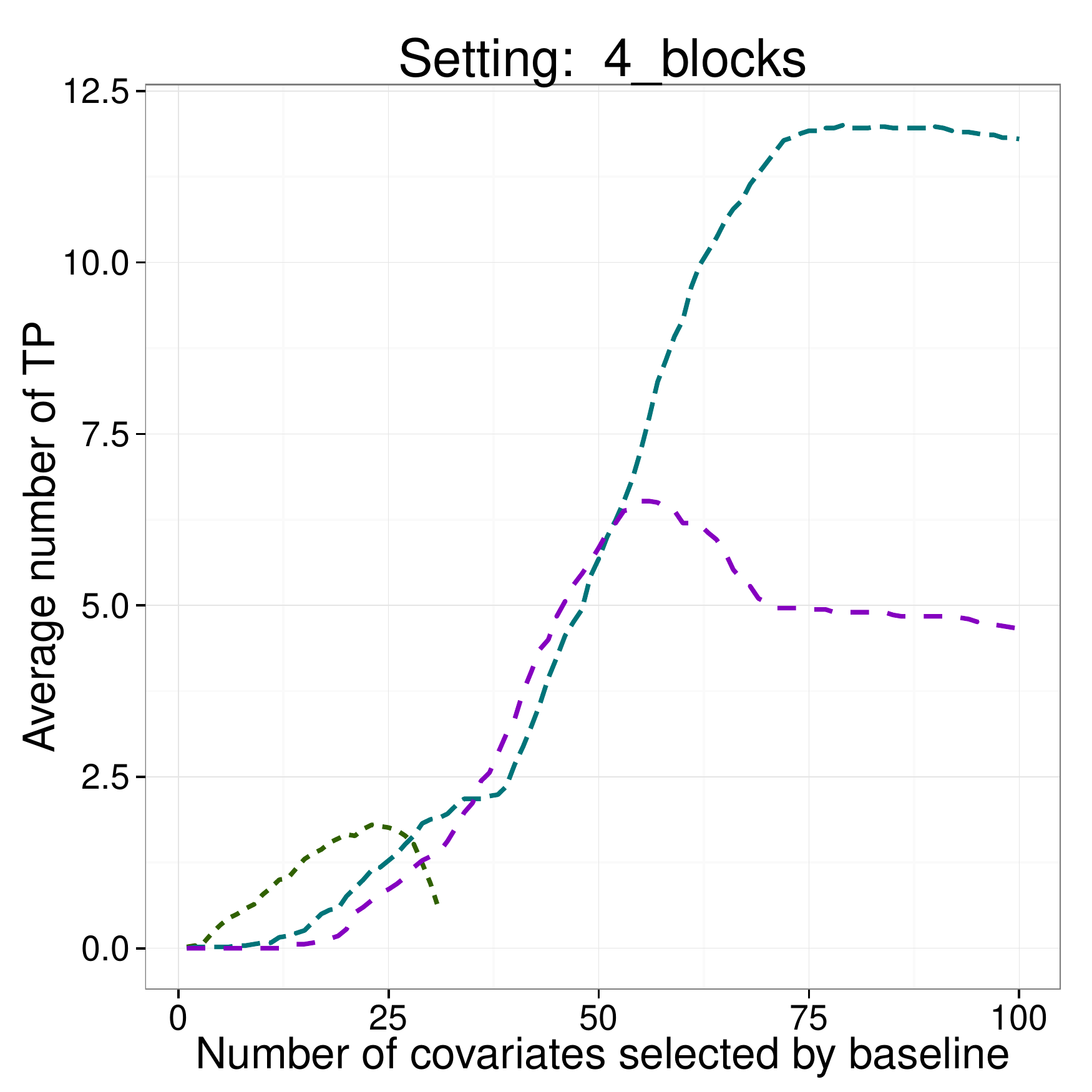}\nolinebreak
  \includegraphics[height=0.29\textwidth]{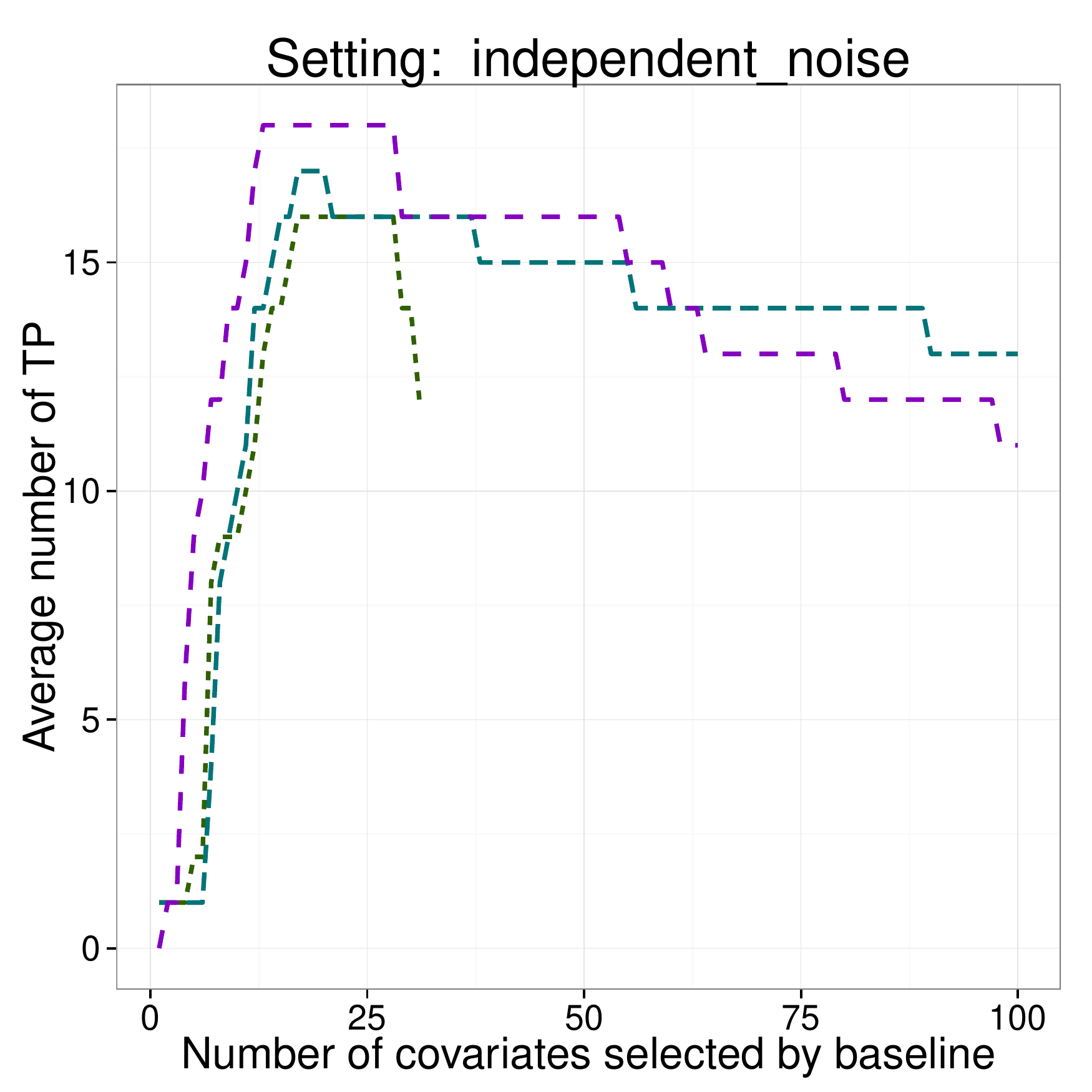}\nolinebreak
  \includegraphics[height=0.29\textwidth]{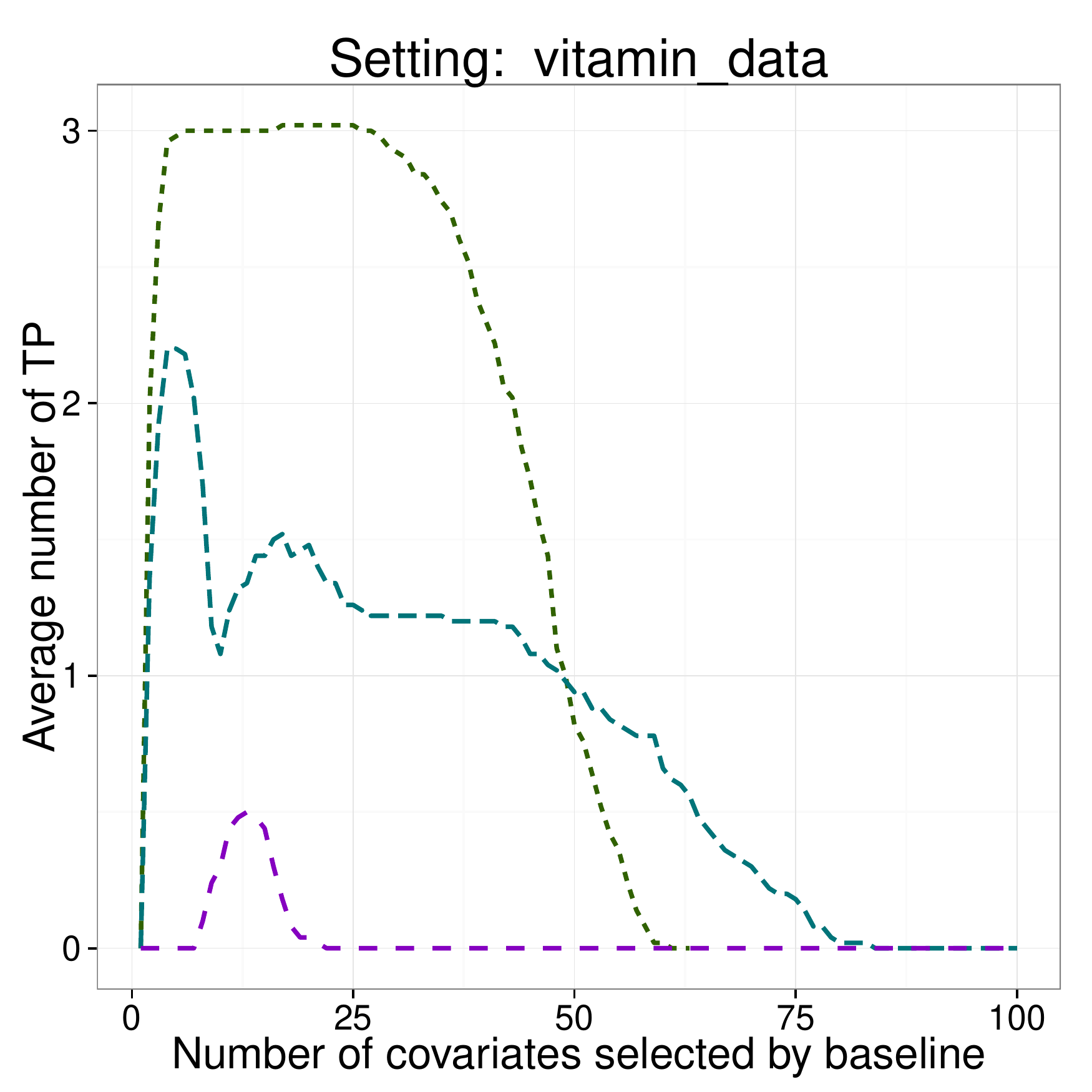}\nolinebreak
  \caption{Like Figure \ref{fig:resultsFp}, except that we plot the average number of true positives.}
  \label{fig:resultsTp}
\end{figure*}

\subsection{Application: Image classification.}\label{prediction}
In this section, we use variable selection as a preprocessing step to solve a classification problem. We investigate whether stabilising the variable selection method with extended stability selection improves the classification error rate compared with applying the variable selection method to the whole dataset.

\paragraph{Dataset Description.}

For prediction we used a subset of 6000 images of handwritten digits from the MNIST  dataset \citep{lecun1998gradient}. Covariates were computed using feature extractors from the 
a collaborative machine learning framework \citep{MASH} where external contributors
can directly submit feature extractors of their choice by uploading C++ code through a web interface.
A total of 48416 covariates were computed on each image from the contributed feature extractors. 
These covariates are heterogeneous because they come from different contributors,
generally exhibit strong correlations, do not have a sparse structure,
and many of them contain at least 
some information about the class label. Here, the goal is not
to strictly identify informative covariates, 
but to reduce the dimensionality of the problem in order to apply a learning and prediction method afterwards. 
This way, the computational complexity and memory requirements of the learning algorithm can be reduced considerably. Note that in these experiments 
we are not aiming at outperforming state-of-the-art
classification results on MNIST, but we wish to evaluate the effect of the proposed extensions
to stability selection on prediction performance, as compared to standard stability
selection with the same base method.

\paragraph{Experimental setup and results.}
As base method, we use CMIM (conditional mutual information maximization,  \citealp{fleuret2004fast}) to select 1000 covariates. As described in Section \ref{sec:basemethods},
Equations \eqref{eq:cmim1} and \eqref{eq:cmim2} CMIM iteratively selects covariates updating in each iteration a score for each covariate. To speed up computation, we use an approximation proposed by \citet{BeiDogBla12b} and only
update scores up until iteration $k$ (recomputation of the scores at each iteration being the
computationally costly part), so that Equation \eqref{eq:cmim2} is replaced by
\[
  \nu(\ell)  = 
  \argmax_{d\in\set{1, \dots, D}\setminus\set{\nu(i), i<\ell}}
  \min_{j \leq \min(\ell-1,k)} \wh{I}\left(X_d;Y|X_{\nu(j)} \right)\,.
  \label{eq:cmim2}
  \]
  In other words, after $k$ normal CMIM iterations 
  the $1000-k$ covariates which remain to be selected are chosen as the ones with the largest values of the score $\min_{j \leq k} \wh{I}\left(X_d;Y|X_{\nu(j)} \right)$.
  We considered $k=10$ and $k=100$ and denoted the resulting stability selection
  method as SFS$_{10}$ or SFS$_{100}$ respectively. Selection by CMIM only (updating scores until iteration 10) is denoted by CMIM$_{10}$.

\begin{table}[b]
  \caption{The effect of the size of covariate subsets on prediction error. Selection of $1000$ covariates with SFS$_k(2,V)$ using CMIM$_k$ as base method; prediction with AdaBoost.MH with various numbers of iterations. Reported are averaged results over 100 repetitions. }
  \label{results_tbl_2}
  \begin{center}
    \begin{tabular}{|r|c|c|c|c|}
      \cline{2-5}
      \multicolumn{1}{c} { }&\multicolumn{ 4}{|c|}{\textbf{SFS$_{10}(2,V), V=$}} \\
      \hline
      \textbf{\# it.} & \textbf{1} & \textbf{2} & \textbf{6} & \textbf{10} \\ \hline
      \textbf{50}  & 8.1 (0.2) & 7.7 (0.3) & 7.5 (0.3)  & \textbf{7.4} (0.3) \\ \hline
      \textbf{100} & 5.0 (0.2) & 4.8 (0.2) & \textbf{4.5} (0.1) & \textbf{4.5} (0.3) \\ \hline
      \textbf{200} & 3.3 (0.1) & 2.9 (0.1) & \textbf{2.5} (0.2) & 2.7 (0.2) \\ \hline
      \textbf{400} & 2.8 (0.1) & 2.7 (0.2) & \textbf{2.4} (0.1) & \textbf{2.4} (0.2) \\ \hline
      \textbf{800} & 2.4 (0.1) & 2.5 (0.2) & 2.2 (0.1) & \textbf{2.1} (0.1) \\ \hline
      \textbf{1600}& 2.0 (0.1) & 2.0 (0.2) & 1.9 (0.2) & \textbf{1.8} (0.1) \\ \hline
      \multicolumn{ 5}{|c|}{\textbf{SFS$_{100}(2,V), V=$}} \\ \hline
      \textbf{50} & 8.0 (0.2) & 7.5 (0.3) & 7.3 (0.2) & \textbf{7.2} (0.3) \\ \hline
      \textbf{100} & 4.9 (0.2) & 4.6 (0.2) & 4.0 (0.3) & \textbf{3.0} (0.2) \\ \hline
      \textbf{200} & 3.3 (0.2) & 3.0 (0.2) & \textbf{2.3} (0.1) & 2.6 (0.2) \\ \hline
      \textbf{400} & 3.0 (0.1) & 2.5 (0.1) & \textbf{2.1} (0.1)  & 2.2 (0.1) \\ \hline
      \textbf{800} & 2.3 (0.1) & 2.2 (0.1) & 2.0 (0.1)  & \textbf{1.9} (0.1) \\ \hline
      \textbf{1600} & 2.1 (0.2) & 1.9 (0.1)  & 1.8 (0.1) & \textbf{1.6} (0.1) \\ \hline
    \end{tabular}
  \end{center}
\end{table}
\begin{table}[b]
  \caption{The effect of the size of observation subsamples on prediction error. Selection of $1000$ covariates with CMIM$_{10}$ only and SFS$_{10}(L,1)$ using CMIM$_{10}$ as base method; prediction with AdaBoost.MH with various numbers of iterations. Reported are averaged results over 100 repetitions.}
  \label{results_tbl_3}
  \begin{center}
    \begin{tabular}{|r|c|c|c|c|}
      \cline{2-5}
      \multicolumn{1}{c} { }& \multicolumn{ 1}{|c|}{\textbf{CMIM$_{10}$}} &\multicolumn{ 3}{|c|}{\textbf{SFS$_{10}(L,1), L=$}} \\ \hline
      \textbf{\# it.} & \textbf{}& \textbf{2} & \textbf{6} & \textbf{10} \\ \hline
      \textbf{50} & 8.4 (0.4) & 8.1 (0.2) & \textbf{8.0} (0.3) & 8.1 (0.2) \\ \hline
      \textbf{100} & 5.5 (0.2) & 5.0 (0.2) & 5.2 (0.3) & \textbf{4.9} (0.2) \\ \hline
      \textbf{200} & 3.6 (0.3) & 3.3 (0.1) & \textbf{3.2} (0.2) & 3.4 (0.2) \\ \hline
      \textbf{400} & 2.8 (0.2) & 2.8 (0.1) & \textbf{2.7} (0.2) & \textbf{2.7} (0.1) \\ \hline
      \textbf{800} & 2.4 (0.1) & 2.4 (0.1) & \textbf{2.3} (0.2) & \textbf{2.3} (0.1) \\ \hline
      \textbf{1600} & 2.1 (0.1) & 2.0 (0.1) & 2.1 (0.1) & \textbf{1.9} (0.1) \\ \hline
    \end{tabular}
  \end{center}
\end{table}

  For all considered  methods  $1000$ covariates are selected at the end,
  corresponding to about $2.1\%$ of the total number of covariates.
Using the selected covariates we applied the AdaBoost.MH algorithm  \citep{escudero2000boosting,schapire1999improved} as a learning and prediction algorithm with various numbers of boosting iterations. 
We report our results in Table~\ref{results_tbl_2} and \ref{results_tbl_3}. In Table~\ref{results_tbl_3} we only show results for SFS$_{10}$, results for SFS$_{100}$ were similar.  The main
conclusions are the following: 
\begin{itemize}
\item We observe a slight trend that $L>2$ leads to better prediction
performance than $L=2$, but it is not statistically significant.
The main conclusion here though, is that taking smaller subsample sizes does not degrade the final performance.
This is of much  relevance in practice, since smaller subsample sizes require less memory.
Furthermore, if the computational complexity $C_{Base}(N)$ of the base method grows  faster than linearly in 
the number of 
observations $N$, the total computing cost is also reduced (since $LC_{Base}(N/L)$ is decreasing in $L$ in that case).
Parallelization is also easier for smaller disjoint subsamples.
\item We observe a statistically significant trend that $V>1$ (smaller covariate subsets)
leads to improvements in prediction performance in comparison to $V=1$ (standard
stability selection).
\end{itemize}

Overall, and paralleling the conclusions of Section~\ref{se:artif}, 
these results demonstrate the relevance of extended stability selection.

\section{Discussion/Conclusion}

We presented theoretical and experimental support for the proposed extensions of the stability
selection methodology.
Concerning subsampling of observations using smaller subsamples,
we generalized error bounds of preceding investigations \citep{meinshausen2010stability,shah2013variable}.
These new bounds give insights into the effect of the subsample size on the selection performance and can provide guidance to practitioners to use a version of stability selection that suits their needs.
Our simulations showed that the obtained false positive bounds can apply in regimes where previous
bounds are void, resulting a more powerful procedure (number of true positives). Still, the obtained
bounds appear loose in practice and could probably be further refined.
Concerning randomising of an arbitrary base procedure by taking random covariate subsets,
we motivated this method from a simplified theoretical toy model, showing in certain circumstances
(heavy tailed score noise) that restricting the search to a random subset increases the probability of correct recovery.
We expect that this second type of extension should be particularly appealing for practitioners in high dimensional settings, where the number of covariates largely exceeds the number of observations.

Experimental results using the precision@20 information retrieval criterion showed that both
extensions improve on standard stability selection whenever the latter improves on the base method applied to the whole dataset.
It remains an open task to determine precise conditions under which a variable selection method can be improved by stability selection. Even though our analysis gave first insights on the dependence of the error probability on the size of observation subsamples, a more precise rule for the optimal choice of the subsample size is left for further work. The same holds for the optimal choice of the size of covariate subsets, even though in practice in very high dimensional problems this choice might be dictated by computational constraints in the first place.

\begin{acknowledgements}
We are extremely grateful to Nicolai Meinshausen and Peter B\"uhlmann for communicating to us the R-code used by \citet{meinshausen2010stability} as well as for numerous discussions. We are indebted to Richard Samworth and Rajen Shah for numerous discussions and for hosting the first author during part of this work. 
We thank Maurilio Gutzeit for helping us with part of the numerical experiments.
\end{acknowledgements}

\appendix


\section{Proofs of theoretical results}

\subsection{Proofs of Section \ref{se:subobs}}

For notational convenience we use the shorthand $S^{\mathrm{base}}(\ell,t)\equiv S^{\mathrm{base}}(X^{(\mathscr{S}(\ell,t))},Y^{(\mathscr{S}(\ell,t))})$\,.
To prove Theorem \ref{th:sampleSplitting} and Corollary \ref{co:sampleSplitting} we need some notation and two lemmas. 
We define
\[
{\Pi}^{\mathrm{simult}}_{L,\ell_0}(d) := \frac{1}{T} \sum_{t=1}^T \mbf{1}\left\lbrace \sum_{\ell=1}^L  \mbf{1}\left\lbrace d \in S_L^{\mathrm{base}}(\ell,t) \right\rbrace  \geq \ell_0 \right\rbrace \]
the ratio of repetitions out of $T$ where covariate $d$ has been selected in at least $\ell_0$ subsamples simultaneously.

\begin{lemma}(Relation of $\Pi^{\mathrm{simult}}$ and $\Pi^{SFS}$)
\label{le:ctrpi}
It holds for any $d\in\cF$:
 \begin{equation*}
  \left( \frac{L-\ell_0+1}{L} \right) \Pi^{\mathrm{simult}}_{L,\ell_0}(d) + \frac{\ell_0-1}{L} \ge {\Pi}^{SFS}_{L}(d)\,.
 \end{equation*}
\end{lemma}

\begin{proof}
We have for all repetitions of drawings of subsamples $t=1,\ldots,T$:
\begin{align*}
&\frac{1}{L} \sum_{\ell=1}^L \ind{d \in S^{\mathrm{base}}(\ell,t)}\\
&\quad\leq \paren{\frac{\ell_{0}-1}{L}}\mbf{1}\left\lbrace \sum_{\ell=1}^L  \mbf{1}\left\lbrace d \in S^{\mathrm{base}}(\ell,t) \right\rbrace  \leq \ell_0 -1 \right\rbrace \\
&\quad+ \mbf{1}\left\lbrace \sum_{\ell=1}^L  \mbf{1}\left\lbrace d \in S^{\mathrm{base}}(\ell,t) \right\rbrace  \geq \ell_0 \right\rbrace\,.
\end{align*}
Averaging over the repetitions $t=1,\ldots,T$\,, we obtain
\begin{align*}
\Pi^{SFS}_{L}(d) 
& \leq \frac{\ell_{0}-1}{L} \left(1 - { {\Pi}^{\mathrm{simult}}_{L,\ell_0}}(d) \right) + {\Pi^{\mathrm{simult}}_{L,\ell_0}}(d)\\
&= \left( \frac{L-\ell_{0}+1}{L} \right) {\Pi^{\mathrm{simult}}_{L,\ell_{0}}(d)} + \frac{\ell_{0}-1}{L}\,.
\end{align*}
\end{proof}

\begin{lemma} \label{le:cher}(Exponential inequality for $\Pi^{\mathrm{simult}}$)
The following inequality holds for any $d\in\cF$,
$\xi>0$, and $\ell_0\in\set{1,\ldots,L}$ such that 
$p_0:= \frac{\ell_0}{L} \geq \pdl$:
 \begin{equation}
\label{eq:ctrprob}
  \prob{ \Pi^{\mathrm{simult}}_{L,\ell_0}(d) \ge\xi}
\leq \frac{1}{\xi} \exp\paren{-L D\paren{p_0, \pdl}}\,.
\end{equation}
\end{lemma}

\begin{proof}
We have
 \begin{align*}
 \e{ {\Pi}^{\mathrm{simult}}_{L,\ell_0}(d)} 
& =  \prob{ \sum_{\ell=1}^L \ind{d \in {S}^{\mathrm{base}}(\ell,1) } \ge \ell_0 } \\
& = \prob{ \mathrm{Bin}\left(L,\pdl \right) \ge \ell_0} \\
& \le \exp\paren{-L D\left(p_0, \pdl\right)}\,.
\end{align*}
The first equality is valid because the $L$ random 
observation subsamples are {\em disjoint}.
Therefore, their joint distribution is the same
as that of $L$ independent samples of size $\left\lfloor \frac{N}{L} \right\rfloor$;
thus $(S^{\mathrm{base}}(\ell,1))_{1\leq \ell \leq L}$ has the same distribution as
$L$ independent copies of the variable $S_L^{\mathrm{base}}$.
The last inequality is the Chernoff binomial bound.
Using Markov's inequality we get \eqref{eq:ctrprob}.
\end{proof}

{\em Proof of Theorem \ref{th:sampleSplitting}.}
We relate $\Pi^{SFS}$ to $\Pi^{\mathrm{simult}}$ and apply an exponential inequality on $\Pi^{\mathrm{simult}}$. For any $d\in A_{\theta,L}$, it holds by definition of $A_{\theta,L}$ and the assumptions on $p_0$ that $\pdl \leq \theta \leq p_0$,
hence it holds by Lemma~\ref{le:ctrpi} and Lemma~\ref{le:cher} that 
\begin{align*}
&\prob{\Pi^{SFS}_{L}(d) \ge \tau}\\
& \quad\le  \prob{\paren{\frac{L-\ell_0+1}{L}}
\Pi^{\mathrm{simult}}_{L,\ell_0}(d) +\frac{\ell_0-1}{L} \ge \tau}\\
& \quad=  \prob{ \Pi^{\mathrm{simult}}_{L,\ell_0}(d) \ge \frac{L\tau-\ell_0+1}{L-\ell_0+1}}\\
& \quad\le \frac{1-p_0+L^{-1}}{\tau-p_0+L^{-1}}
\exp\paren{-L D\paren{p_0, \pdl}}\,,
\end{align*}
where we have used $\xi:=\frac{L\tau-\ell_0+1}{L-\ell_0+1}$\,. This result generalizes  \citet[Lemma 5]{shah2013variable}. Hence 
\begin{align*}
&\frac{\e{\abs{S^{SFS}_{L,\tau} \cap A_{\theta,L}}}}{\abs{A_{\theta,L}}}\\
& \quad =  \frac{1}{\abs{A_{\theta,L}}} \sum_{d\in A_{\theta,L}}  
\prob{\Pi^{SFS}_{L,\ell_0}(d) \ge \tau} \\
& \quad\leq \frac{1-p_0+L^{-1}}{\tau-p_0+L^{-1}}
\frac{1}{\abs{A_{\theta,L}}} \sum_{d\in A_{\theta,L}}
\exp\paren{-L D\paren{p_0, \pdl}}\,.
\end{align*}
Since $x\rightarrow \exp ( -L D(p_0,x))$ is non-decreasing, we obtain the first part of the result by upper bounding for all 
$d \in A_{\theta,L}$:
\[
\exp\paren{-L D\paren{p_0, \pdl}} \leq
\exp\paren{-L D\paren{p_0, \theta}}\,.
\]
For the second part, we use the upper bound
\begin{align*}
\exp\paren{-L D\paren{p_0, \pdl}} & =
\frac{\exp\paren{-L D\paren{p_0, \pdl}}}{\pdl}
\pdl\\
& \leq \frac{\exp\paren{-L D\paren{p_0, \theta}}}{\theta}
\pdl\,,
\end{align*}
since the function $x\mapsto \frac{\exp\paren{-L D\paren{p_0, x}}}{x}$ can be shown to be non-decreasing for 
$x \leq p_0 - L^{-1}$\,. Finally, summing over $d \in A_{\theta,L}$, observe
\[
\sum_{d\in A_{\theta,L}} p_L(d)
= \e{ \sum_{d\in A_{\theta,L}} \ind{ d \in S^{\mathrm{base}}_L}} = \e{\abs{A_{\theta,L} \cap S^{\mathrm{base}}_L}},
\]
leading to the desired conclusion.
Equations \eqref{eq3} and \eqref{eq4} can be proved similarly.
\qed\\

Proof of Corollary \ref{co:sampleSplitting}. This follows the same argument as in \citet{shah2013variable}.
If the variable selection was completely at random, the marginal selection probability of any given covariate would be $\frac{q_L }{D}$, where we recall $q_L =\e{\abs{S^{\mathrm{base}}_L}}$ is the average number of covariates selected by the base method. 
As we assume that the selection probability of a signal covariate is better than random; it entails that for any $d \in \cN^C$, we must have $\pdl > \frac{q_L }{D}$. Conversely, as all noise covariates have the same probability to be selected by the base method, one has $\pdl < \frac{q_L }{D}$ for any $d \in \cN$. Therefore, with $\theta:=\frac{q_L }{D}$ we must have $A_{\theta,L}=\cN$ and $A_{\theta,L}^c=\cN^C$. Inequality \eqref{eq1} therefore implies \eqref{eq:MB-like}, wherein we have taken a minimum over the
range of $\ell_0$ allowed in Theorem \ref{th:sampleSplitting}.
\qed

\subsection{Proofs for Section \ref{se:subcov}}

\begin{proof}[Proof of Theorem \ref{th:frechet}]
  We can first bound the error probability from above by omitting $Q_d$:
  \begin{align*}
    \mathbb{P}\left[ \hat{d}_D \in A_{D,\theta} \right]   &=  \mathbb{P}\left[ \max_{d \in A_{D,\theta}} \hat{Q}_d > \max_{d \in A_{D,\theta}^c}\hat{Q}_d \right]\\
    &=\mathbb{P}\left[ \max_{d \in A_{D,\theta}} \left(Q_d + \eps_d\right) > \max_{d \in A_{D,\theta}^c}\left(Q_d + \eps_d\right) \right]\\
    &\le \mathbb{P}\left[ \max_{d \in A_{D,\theta}} (\theta + \eps_d) > \max_{d \in A_{D,\theta}^C} (\theta +\eps_d) \right]\\
& = \prob{ \argmax_{d\in\set{1,\ldots,D}} \eps_d \in A_{D,\theta}} = \frac{\abs{A_{D,\theta}}}{D}  \rightarrow \eta,
  \end{align*}
  as $D\rightarrow \infty$. If $\eta=0$, the conclusion is therefore established; in the remainder of
the proof we hence assume $\eta>0$. We defer to the end of the proof the case $\eta=1$ and assume 
for now that $\eta\in(0,1)$.  Then $\frac{|A_{D,\theta}|}{D} \rightarrow \eta\in(0,1)$ implies
both $|A_{D,\theta}| \rightarrow \infty$ and $|A^c_{D,\theta}|\rightarrow \infty$, as well as
$\frac{|A_{D,\theta}^c|}{|A_{D,\theta}|} \rightarrow \gamma := \frac{1-\eta}{\eta}$.
We return to the error probability and bound it from below by using  $Q_d\geq 0$ for 
$d \in A_{D,\theta}$ and $Q_d\leq M$ for $d\in A_{D,\theta}^c$:
  \begin{equation}
\label{eq:firststep}
    \mathbb{P}\left[\hat{d}_D \in A_{D,\theta} \right] 
    \ge  \mathbb{P}\left[ \max_{d \in A_{D,\theta}} \eps_d 
      > M + \max_{d \in A_{D,\theta}^C}\eps_d \right].
\end{equation}
Since the distribution of $\eps_i$ belongs to MDA($\fre(\alpha)$), from classical results
of extreme value theory \citep[Theorem 3.3.7]{embrechts1997modelling} we know that there exists a slow varying function $L$ so that, if we denote $G(x):=x^{1/\alpha}L(x)$, then
  \begin{align}\label{eq:Frechet}
    &\frac{\max_{d \in A_{D,\theta}} \eps_d }{G(|A_{D,\theta}|)} \rightarrow \fre(\alpha)\nonumber\\
    &\; \text{  and } \;\nonumber\\ 
    &\frac{\max_{d \in A_{D,\theta}^c} \eps_d }{G(|A_{D,\theta}^c|)} \rightarrow \fre(\alpha),
  \end{align}
in the sense of convergence in distribution, as $D\rightarrow\infty$.
Following on \eqref{eq:firststep}:
\begin{multline*}
    \mathbb{P}\left[\hat{d}_D \in A_{D,\theta} \right]  \ge  \mathbb{P}\left[ \frac{\max_{d \in A_{D,\theta}} \eps_d}{G(|A_{D,\theta}|)} > \frac{M}{G(|A_{D,\theta}|)} \right. \\
    \left.+\frac{G(|A^c_{D,\theta}|)}{G(|A_{D,\theta}|)} \frac{\max_{d \in A_{D,\theta}^c}\eps_d}{G(|A^c_{D,\theta}|)}  \right].
  \end{multline*}
As $L$ is slowly varying, we have $\lim_{x \rightarrow \infty} \frac{L(a x)}{L(x)} \rightarrow 1$ uniformly for $a$ belonging to a bounded interval of the positive real axis \citep[Theorem A 3.2]{embrechts1997modelling}. We deduce
  \begin{align}\label{eq:ToGamma}
  \frac{G(|A^c_{D,\theta}|)}{G(|A_{D,\theta}|)} = \paren{\frac{|A_{D,\theta}^c|}{|A_{D,\theta}|}}^{\frac{1}{\alpha}}
\frac{L\paren{|A_{D,\theta}|\paren{\frac{|A^c_{D,\theta}|}{|A_{D,\theta}|}}}}{L(|A_{D,\theta}|)} \rightarrow \gamma^{1/\alpha},
  \end{align}
as $D\rightarrow\infty$.   We apply Slutsky's theorem \citep[Example A 2.7]{embrechts1997modelling} to Equations (\ref{eq:ToGamma}) and (\ref{eq:Frechet}) to obtain
  \[
\frac{G(|A^c_{D,\theta}|)}{G(|A_{D,\theta}|)} \frac{\max_{d \in A_{D,\theta}^c}\eps_d}{G(|A^c_{D,\theta}|)} 
\rightarrow \fre(\alpha,  \gamma^{1/\alpha})
  \]
in distribution, where $\fre(\alpha,c)$ denotes the $\fre(\alpha)$ distribution rescaled by a factor $c>0$.

Further, slow variation of $L$ implies that $L(x)$ is asymptotically negligible with respect to any power function, so that 
  \begin{equation}
\label{eq:Mterm}
\frac{M}{G(|A_{D,\theta}|)} = \frac{M}{|A_{D,\theta}|^{\frac{1}{\alpha}}L(|A_{D,\theta}|)} 
\rightarrow 0, \;\; \text{ as } \;\; D\rightarrow \infty.
  \end{equation}
  As the maxima in $\max_{d \in A_{D,\theta}} \eps_d$ and $\max_{d \in A_{D,\theta}^C}\eps_d$ are taken over disjoint sets of independent random variables, they are independent. Since they converge marginally in distribution, they also converge jointly and their difference converges due to the continuous mapping theorem \citep[Theorem A 2.6]{embrechts1997modelling}.
Combining with \eqref{eq:Mterm} and using Slutsky's theorem again, we conclude that
  \begin{align*}\label{event}
 \frac{\max_{d \in A_{D,\theta}} \eps_d}{G(|A_{D,\theta}|)} 
      - \frac{M}{G(|A_{D,\theta}|)} -
\frac{G(|A^c_{D,\theta}|)}{G(|A_{D,\theta}|)} \frac{\max_{d \in A_{D,\theta}^c}\eps_d}{G(|A^c_{D,\theta}|)}
\end{align*}
converges in distribution  to the difference of two independent Fr\'echet distributed random variables. This convergence implies the convergence of the c.d.f. for all continuity points \citep[Equation A.1]{embrechts1997modelling}. As the limiting distribution is continuous, we finally obtain
\[
\liminf_{D\rightarrow \infty} \mathbb{P}\left[\hat{d}_D \in A_{D,\theta} \right]
\geq \mathbb{P}\left[ F - \gamma^{\frac{1}{\alpha}} F' > 0 \right],
\]
where $F,F'$ are independent $\fre(\alpha)$ random variables.
To identify the value of this lower bound, observe that it 
is also the limiting value of
\[
\prob{ \argmax_{d\in\set{1,\ldots,D}} \eps_d \in A_{D,\theta}} = \frac{|A_{D,\theta}|}{D}.
\]
Indeed, it suffices to repeat the above argument, except for skipping
inequality \eqref{eq:firststep}. Hence this limiting value is exactly equal to $\eta$.

Finally, for the case $\eta=1$, observe that the above argument remains valid
provided $|A_{D,\theta}^c|\rightarrow \infty$. Even if this is not the case (i.e. 
$|A_{D,\theta}^c|$ remains bounded), then the conclusion would be {\em a fortiori} true
since we could replace $A_{D,\theta}^c$ by a slightly larger set of cardinality $\ln(D)$
(say), which can only decrease the lower bound while still obtaining the above
limiting value.
\end{proof}

\begin{proof}[Proof of Theorem \ref{th:gaussian}]
To show the convergence of the error probability we use  similar arguments as in the proof of Theorem \ref{th:frechet}. From classical results of extreme value
theory for independent standard normal random variables $(\zeta_k)_{k \in \mbn}$ \citep[Example 3.3.29]{embrechts1997modelling} it holds that
\[
\frac{\max_{i \leq k} \zeta_i - b_k }{a_k} \rightarrow \mathsf{Gumbel}\,, \qquad 
\text{ as } k \rightarrow \infty\,,
\]
in distribution, where $a_k:= \frac{1}{\sqrt{2\ln(k)}}$ and $b_k:=\sqrt{2\ln(k)} - \frac{\ln(4\pi)+\ln(\ln(k))}{2\sqrt{2\ln(k)}}$. Below, to clarify the argument we
will introduce
$(\zeta_k)_{k \in \mbn}$ and $(\zeta'_k)_{k \in \mbn}$ two independent sequences
of independent standard normal variables.

Denote $k_D:=|A_{D,\theta'}|$, $\ell_D:= |A_{D,\theta}^c|$ and $\Delta:=\theta-\theta'>0$. Since $A_{D,\theta}^c \subseteq A_{D,\theta'}^c$
we can bound the error probability from above as follows:
\begin{align*}
  \mathbb{P}\left[ \hat{d}_D \in A_{D,\theta'} \right]
& \leq \mathbb{P}\left[ \max_{d \in A_{D,\theta'}} \left(Q_d + \eps_d\right) > \max_{d \in A_{D,\theta}^c}\left(Q_d + \eps_d\right) \right]\\
  & \leq \prob{\max_{d \in A_{D,\theta'}} \eps_d > \max_{d \in A_{D,\theta}^c} \eps_d + \Delta}\\
  & = \prob{\max_{d\leq k_D} \zeta'_d > \max_{d \leq \ell_D} \zeta_d + \Delta}\,.
\end{align*}
The last  equality holds since $A_{D,\theta'}$
and $A_{D,\theta}^c$ are disjoint;
it is purely formal but notationally convenient for the sequel.
Now denote $k'_D := \max(k_D,\ell_D)$\,. The above implies
\begin{multline*}
  \mathbb{P}\left[ \hat{d}_D \in A_{D,\theta'} \right]\\
  \begin{aligned}
  & \leq \prob{\max_{d\leq k'_D} \zeta'_d > \max_{d \leq \ell_D} \zeta_d + \Delta}\\
  & = \mathbb{P}\left[ 
\frac{a_{k'_D}}{a_{\ell_D}}\paren{\frac{\max_{d \leq k'_D} \zeta'_d - b_{k'_D}}{a_{k'_D}}} - \frac{\max_{d \leq \ell_D}\zeta_d - b_{\ell_D}}{a_{\ell_D}} \right.\\
 & \qquad \quad> \left.\frac{b_{\ell_D} - b_{k'_D} + \Delta}{a_{\ell_D}} \right]
  \end{aligned}
  \end{multline*}
We treat the different terms in the above upper bound. We have
\[
 \frac{a_{k'_D}}{a_{\ell_D}} 
=   \frac{\sqrt{2\ln(\ell_D)}}{\sqrt{2\ln(k'_D)}} 
 \leq  1 .
\]
Noting that $b_k = \sqrt{2 \ln k} + o(1)$, 
we have
\begin{align}
&  \frac{b_{\ell_D} - b_{k'_D} + \Delta}{a_{\ell_D}} \nonumber \\
& \qquad= \sqrt{2 \ln \ell_D} \paren{ \sqrt{2 \ln \ell_D} - \sqrt{2 \ln k'_D} + \Delta + o(1)} \nonumber  \\ 
&\qquad= -2\paren{\ln \frac{k'_D}{\ell_D}} \paren{\frac{\sqrt{\ln \ell_D}}{\sqrt{\ln k'_D} + \sqrt{\ln \ell_D}}} \label{eq:plast} \\ 
 & \qquad \quad + \Delta \sqrt{2\ln \ell_D} + o(\sqrt{\ln \ell_D}) \nonumber  \\
 & \qquad \geq \Delta \sqrt{2\ln \ell_D} + o(\sqrt{\ln \ell_D}).  \label{eq:lastInequality}
\end{align}
To check that the last inequality holds, note that Assumption (ii) of the Theorem states that $\liminf \frac{\ell_D}{D} := \eta >0$,
in particular $\ell_D \rightarrow \infty$. On the other hand, since $A_{D,\theta'} \subseteq A_{D,\theta}$, we have $\limsup \frac{k_D}{D} \leq 1-\eta$, therefore $\limsup \frac{k_D}{\ell_D} \leq \frac{1-\eta}{\eta}:=\gamma$ and finally
$\limsup \frac{k'_D}{\ell_D} \leq \max(\gamma,1)$\,.
Since $\ln \frac{k'_D}{\ell_D} \geq 0$\,,
$ \limsup \ln \frac{k'_D}{\ell_D} \leq \paren{\ln \gamma}_+$, and the second factor in \eqref{eq:plast} is positive and upper bounded by 1,
the whole term in \eqref{eq:plast} is $O(1)$\,, so that inequality \eqref{eq:lastInequality}  follows.

We deduce that for any $B>0$ and for $D$ large enough 
$\frac{b_{\ell_D} - b_{k'_D}+\Delta}{a_{\ell_D}} >B$ holds, and we have 
\begin{multline*}
  \mathbb{P}\left[\hat{d}_D \in A_{D,\theta'} \right]
  \leq \mathbb{P}\left[  \paren{\frac{\max_{d \leq k'_D} \zeta'_d - b_{k'_D}}{a_{k'_D}}} \right.\\
\left. \qquad - \frac{\max_{d \leq \ell_D}\zeta_d - b_{\ell_D}}{a_{\ell_D}} 
> B \right].
\end{multline*}

By similar arguments as in the proof of Theorem \ref{th:frechet},
the latter upper bound converges to $\mathbb{P}[ G - G'>B ]$, where $G,G'$
are two independent Gumbel random variables. As $B$ is arbitrary we come to the
announced conclusion.
\end{proof}

\bibliographystyle{abbrvnat}
\bibliography{bibliography}   


\checknbdrafts

\end{document}